\documentclass[11pt,onecolumn]{IEEEtran}
\usepackage[dvipsnames]{xcolor}
\usepackage[noadjust]{cite}
\usepackage[pdftex]{graphicx}
\graphicspath{{figs/pdf/}}
\usepackage{amsmath}
\usepackage{amsfonts}
\usepackage{array}
\usepackage{url}
\usepackage{hyperref}
\usepackage{amsthm}
\usepackage[capitalise,noabbrev]{cleveref}
\usepackage{xspace}

\newtheorem{theorem}{Theorem}
\newtheorem{lemma}[theorem]{Lemma}
\newtheorem{prop}[theorem]{Proposition}

\theoremstyle{definition}
\newtheorem{definition}{Definition}
\newtheorem{example}{Example}
\theoremstyle{remark}
\newtheorem{remark}{Remark}
\newcommand{\etal}{et al.\xspace}
\newcommand{\ie}{i.e.,\xspace}

\newcommand{\CodeName}{Convertible Code\xspace}
\newcommand{\CodeNames}{\CodeName{s}\xspace}
\newcommand{\Codename}{Convertible code\xspace}

\newcommand{\codename}{convertible code\xspace}
\newcommand{\codenames}{\codename{}s\xspace}
\newcommand{\mdscodenames}{MDS convertible codes\xspace}
\newcommand{\regime}{merge regime\xspace}
\newcommand{\conversion}{conversion\xspace}
\newcommand{\codeconversion}{code conversion\xspace}
\newcommand{\codeconversions}{\codeconversion{}s\xspace}
\newcommand{\convert}{convert\xspace}
\newcommand{\converted}{converted\xspace}
\newcommand{\conversions}{\conversion{}s\xspace}
\newcommand{\cost}{access cost\xspace}
\newcommand{\optimal}{access-optimal\xspace}
\newcommand{\stripe}{stripe\xspace}
\newcommand{\encvec}{encoding vector\xspace}
\newcommand{\encvecs}{\encvec{}s\xspace}
\newcommand{\block}{block\xspace}
\newcommand{\blocks}{\block{}s\xspace}
\newcommand{\devices}{disks\xspace}
\providecommand{\adots}{\text{\reflectbox{\(\ddots\)}}}
\newcommand{\IdMat}{\mathbf{I}}
\newcommand{\V}[1]{\mathbf{#1}}
\newcommand{\field}[1]{\mathbb{F}_{#1}}
\newcommand{\Code}{\mathcal{C}}
\newcommand{\initial}{I}
\newcommand{\final}{F}
\newcommand{\ICode}{\Code^{\initial}}
\newcommand{\FCode}{\Code^{\final}}

\newcommand{\IGen}{\mathbf{G}^{\initial}}
\newcommand{\IGenp}{\tilde{\mathbf{G}}^{\initial}}
\newcommand{\FGen}{\mathbf{G}^{\final}}
\newcommand{\IPar}{\mathbf{P}^{\initial}}
\newcommand{\FPar}{\mathbf{P}^{\final}}
\newcommand{\ICol}[1]{\V{g}^{\initial}_{#1}}
\newcommand{\FCol}[1]{\V{g}^{\final}_{#1}}
\newcommand{\IColp}[2]{\V{\tilde{g}}^{\initial}_{{#1},{#2}}}
\newcommand{\Down}{\mathcal{D}}
\newcommand{\Vectors}{\mathcal{V}}
\newcommand{\Vectorss}{\mathcal{W}}

\newcommand{\StripeIndexSet}[1]{\mathcal{K}_{#1}}
\newcommand{\StripeICols}[1]{\mathcal{S}^{\initial}_{#1}}
\newcommand{\AllICols}{\mathcal{S}^{\initial}}
\newcommand{\AllFCols}{\mathcal{S}^{\final}}
\newcommand{\NewCols}{\mathcal{N}}

\newcommand{\UnchangedStripeCols}[1]{\mathcal{U}_{#1}}
\newcommand{\ParamFormat}[4]{({#1},{#2};{#3},{#4})}
\newcommand{\Params}[4]{\ParamFormat{\In = {#1}}{\Ik = {#2}}{\Fn = {#3}}{\Fk = {#4}}}
\newcommand{\ParamsDefault}{\ParamFormat{\In}{\Ik}{\Fn}{\Fk}}
\newcommand{\ParamCode}[4]{\(\Params{#1}{#2}{#3}{#4}\) \codename}
\newcommand{\CodeDefault}{\(\ParamsDefault\) \codename}
\newcommand{\RegimeCodeDefault}{\(\ParamFormat{\In}{\Ik}{\Fn}{\Fk = \Cs\Ik}\) \codename}
\newcommand{\RegimeCodesDefault}{\RegimeCodeDefault{}s}
\newcommand{\DownSize}{d}
\newcommand{\MinDownSizeDefault}{\DownSize^{*}\ParamsDefault}
\newcommand{\Partition}{\mathcal{P}}
\newcommand{\IPart}{\Partition_{\initial}}
\newcommand{\FPart}{\Partition_{\final}}
\newcommand{\In}{n^{\initial}}
\newcommand{\Ik}{k^{\initial}}
\newcommand{\Ir}{r^{\initial}}
\newcommand{\Fn}{n^{\final}}
\newcommand{\Fk}{k^{\final}}
\newcommand{\Fr}{r^{\final}}
\newcommand{\NNodes}{M}
\newcommand{\Transition}[2]{T_{{#1}\!\to{#2}}}
\newcommand{\Cs}{\lambda}
\newcommand{\Msg}{\V{m}}
\newcommand{\Constructible}[1]{\ensuremath{#1}-column constructible}
\newcommand{\BlockConstructible}[1]{\ensuremath{#1}-column block-constructible}
\newcommand{\hankelone}{Hankel-I}
\newcommand{\hankeltwo}{Hankel-II}
\DeclareMathOperator{\Perm}{Perm}
\DeclareMathOperator{\Sign}{sgn}
\DeclareMathOperator{\Id}{\sigma_{\mathrm{id}}}
\DeclareMathOperator{\Proj}{proj}

\DeclareMathOperator{\Cols}{cols}
\DeclareMathOperator{\Lcm}{lcm}
\DeclareMathOperator{\Mod}{mod}
\DeclareMathOperator{\Diag}{diag}

\begin{document}

\title{\CodeNames: Efficient Conversion of Coded Data in Distributed Storage}

\author{\IEEEauthorblockN{Francisco Maturana and K. V. Rashmi}

\IEEEauthorblockA{Computer Science Department\\
Carnegie Mellon University\\
\{fmaturan, rvinayak\}@cs.cmu.edu}}

\maketitle

\begin{abstract}
    Large-scale distributed storage systems typically use erasure codes to provide durability of data in the face of failures.
    A set of \(k\) blocks to be stored is encoded using an \([n, k]\) code to generate \(n\) blocks that are then stored on different storage nodes.
    The redundancy configuration (that is, the parameters \(n\) and \(k\)) is chosen based on the failure rates of storage devices, and is typically kept constant.
    However, a recent work by Kadekodi et al.\ shows that the failure rate of storage devices vary significantly over time, and that adapting the redundancy configuration in response to such variations provides significant benefits: a \(11\%\) to \(44\%\) reduction in storage space requirement, which translates to enormous amounts of savings in resources and energy in large-scale storage systems.
    However, converting the redundancy configuration of already encoded data by simply re-encoding (the default approach) requires significant overhead on system resources such as accesses, device IO, network bandwidth, and compute cycles.
    
    In this work, we first present a framework to formalize the notion of \emph{\codeconversion}---the process of converting data encoded with an \([\In, \Ik]\) code into data encoded with an \([\Fn, \Fk]\) code 
    while maintaining desired decodability properties, such as the maximum-distance-separable (MDS) property.
    We then introduce \emph{\codenames}, a new class of codes that allow for \codeconversions in a resource-efficient manner.
    For an important parameter regime (which we call the \regime) along with the widely used linearity and MDS decodability constraint, we prove tight bounds on the number of nodes accessed during \codeconversion.
    In particular, our achievability result is an explicit construction of MDS \codenames that are optimal for all parameter values in the \regime\ albeit with a high field size. 
    We then present explicit low-field-size constructions of optimal MDS \codenames\ for a broad range of parameters in the \regime.
    Our results thus show that it is indeed possible to achieve \codeconversions\ with significantly lesser resources as compared to the default approach of re-encoding.
\end{abstract}

\section{Introduction}\label{sec:intro}
Large-scale distributed storage systems form the bedrock of modern data processing systems. 
Such storage systems comprise hundreds of thousands of storage devices and routinely face failures in their day-to-day operation~\cite{ford2010availability,rashmi2013hotstorage,rashmi2014hitchhiker,asterisxoring}. 
In order to provide resiliency against such failures, storage systems employ redundancy, typically in the form of erasure codes~\cite{ghemawat2003google,facebookECsavings2010_forACM,huang2012erasure,hadoophdfsec}.
Under erasure coding, a set of \(k\) data blocks to be stored is encoded using an \([n, \ k]\) code to generate \(n\) coded blocks.
A set of \(n\) encoded blocks that correspond to the same \(k\) original data blocks is called a ``\stripe''.
Each of the \(n\) coded blocks in a stripe is stored on a different storage node (typically chosen from different failure domains).
The amount of redundancy added using an erasure code is a function of the redundancy configuration, that is, parameters \(n\) and \(k\).
These parameters are chosen so as to achieve predetermined thresholds on reliability and availability, such as the mean-time-to-data-loss (MTTDL).

The key factor that determines MTTDL for chosen parameters is the failure rate of the storage devices
In a recent work~\cite{HEART}, Kadekodi \etal show that failure rates of storage devices in large-scle storage systems vary significantly over time (for example, by \textit{more than 3.5-fold} for certain disk families).
Thus, it is advantageous to change the redundancy configuration in response to such variations
Kadekodi \etal~\cite{HEART} present a case for tailoring erasure code parameters to the observed failure rates and show that an \(11\%\) to \(44\%\) reduction in storage space can be achieved by adapting the redundancy configuration according to the changing failure rates.
Such a reduction in storage space requirement translates to significant savings in the cost of resources and energy consumed in large-scale storage systems.

In particular, disk failure rates exhibit a \emph{bathtub curve} during the lifetime of disks, which is characterized by three phases: \emph{infancy}, \emph{useful life}, and \emph{wearout}, in that order~\cite{HEART}.
Disk failure rate during infancy and wearout can be multiple times higher than during useful life.
As a consequence, the chosen redundancy setting will likely be too high for some periods, which is a waste of resources, and too low for other periods, which increases the risk of data loss.
Kadekodi \etal~\cite{HEART} address this problem by changing the code rate (that is, the parameters of the erasure coding scheme) as the devices go through different phases of life. 
For example, given a group of nodes with certain failure characteristics, the system may use a \([14, 10]\) code during infancy, then \convert\ to a \([24, 20]\) code during useful life, and finally \convert\ back to a \([14, 10]\) code during wearout.
We refer the reader to~\cite{HEART} for an in-depth study on failure rate variations and the advantages of adapting the erasure-code parameters with these variations.   

Adapting the redundancy configuration requires modifying the code rate for all the \stripe{s} that have at least one block stored on a certain disk group when the failure rate of that disk group changes by more than a threshold amount~\cite{HEART}. 
Changing the code rate, that is the parameters of the erasure code, employed on already encoded data can be highly resource intensive, potentially requiring to access multiple storage devices, read large amounts of data, transfer it over the network, and re-encode it.
Modifying the code parameters using the default approach requires reading at least \(k\) blocks from each stripe, transferring over the network and re-encoding.
In large-scale storage systems, \devices are deployed in large batches, and hence a large number of \devices go through failure-rate transitions concurrently.
Thus, adapting redundancy configuration by using the default approach of re-encoding generates highly varying and prohibitively large load spikes, which adversely affect the foreground traffic.
This places significant burden on precious cluster resources such as accesses, disk IO, network bandwidth, and computation cycles (CPU). 
Furthermore, in some cases these conversions need be performed urgently, such as the case where there is an unexpected rise in failure rates and \conversion\ is necessary to reduce the risk of data loss.
In such cases, it is necessary to be able to perform fast conversions.
Motivated by these applications, in this paper, \emph{we initiate a formal study of such \codeconversions} by exploring the following questions:
\begin{itemize}
    \item What are the fundamental limits on resource consumption of \codeconversions?
    \item How can one design codes that efficiently facilitate \codeconversions?
\end{itemize}

Formally, the goal is to \convert\ data that is already encoded using an \([\In, \Ik]\) code (denoted by \(\ICode\)) into data encoded using an \([\Fn, \Fk]\) code (denoted by \(\FCode\))\footnote{The superscripts \(\initial\) and \(\final\) stand for initial and final respectively, representing the initial and final state of the \conversion.}, with desired constraints on decodability such as both initial and final codes satisfying the maximum-distance-separable (MDS) property. 
Clearly, it is always possible to read the original data (and decode if needed) and re-encode according to \(\FCode\).
However, such a re-encoding approach requires accessing several nodes (\(\Ik\) nodes per stripe for MDS codes), reading out all the data, transferring over the network, and re-encoding, which consumes large amounts of access, disk IO, network bandwidth, and CPU resources.

The question then is whether one can perform such \conversions\ in a more resource-efficient manner, while satisfying the decodability constraints.
We now present an example showing how resource-efficient \conversion\ can be achieved in a simple manner for certain parameters.

\begin{example}
Consider \(\In = \Ik + 1,\, \Fn = \Fk + 1\), and \(\Fk = 2\Ik\), with the requirement that both \(\ICode\) and \(\FCode\) are MDS.
This \conversion\ can be achieved by ``merging'' two stripes of the initial code into one stripe, for each stripe of the final code.
Let us focus on the number of \blocks\ accessed during \conversion.
Using the default approach of re-encoding to achieve the \conversion\ requires accessing \(\Ik\) \blocks\ from two stripes of encoded data under \(\ICode\) (initial stripes) to create one stripe of encoded data under \(\FCode\) (final stripe).
That is, each stripe of encoded data under the final code \(\FCode\) requires accessing \(2\Ik\) \blocks.
Alternatively, as depicted in \cref{fig:combine_ex}, one can choose \(\ICode\) and \(\FCode\) to be systematic, single-parity-check codes, with the parity \block\ holding the XOR of the data \blocks\ in each stripe (shown with a shaded box in the figure).
To \convert from \(\ICode\) to \(\FCode\), one can compute the XOR between the single parity in each stripe, and store the result as the parity \block\ for the stripe under \(\FCode\).
This alternative approach requires accessing only two \blocks\ for each final stripe, and thus is significantly more efficient in the number of accessed \blocks\ as compared to the default approach.
\end{example}
\begin{figure}
    \centering
    \includegraphics[width=.5\textwidth]{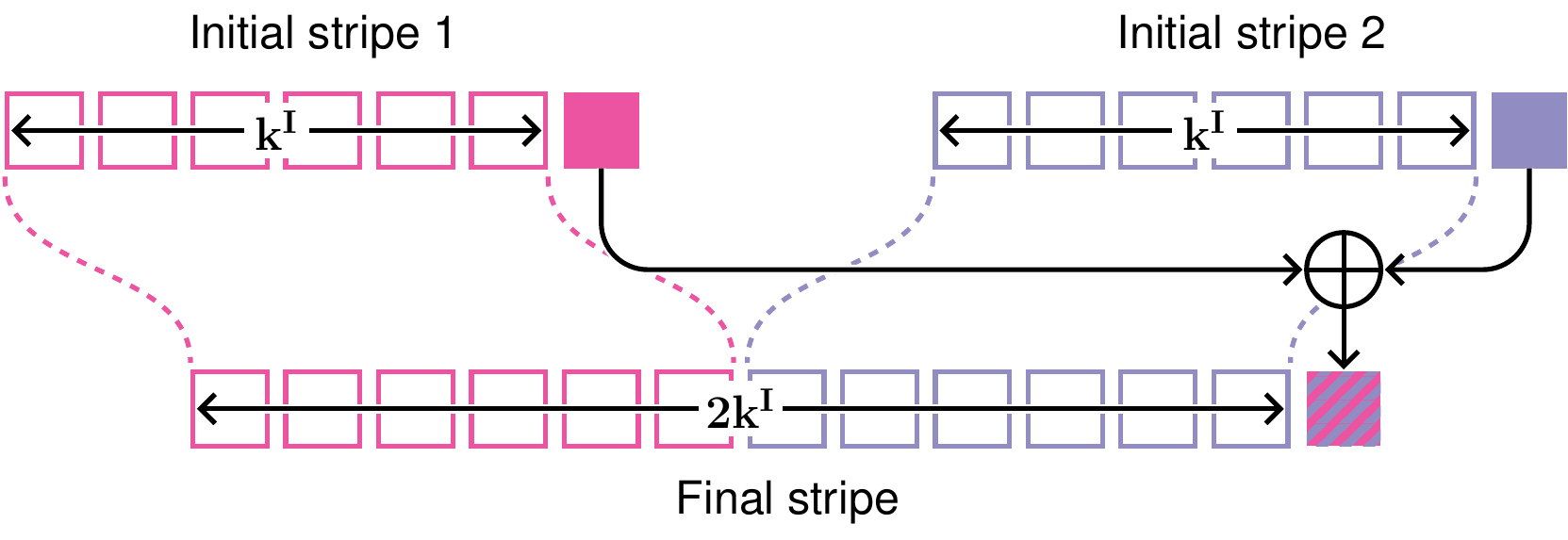}
    \caption{
    Example of \codeconversion\ without re-encoding: two stripes of a \([\Ik + 1, \Ik]\) single-parity-check code become one stripe of a \([2\Ik + 1, 2\Ik]\) single-parity-check code.
    The parity blocks are shown shaded.
    The data blocks from each stripe are preserved, and the single parity block from the final stripe is computed as the XOR of the parities from each of the initial stripes.
    }
    \label{fig:combine_ex}
\end{figure}

In this paper, we first propose a novel framework that formalizes the concept of \emph{\codeconversion}, that is, the process of converting data encoded with an \([\In, \Ik]\) code into data encoded with an \([\Fn, \Fk]\) code while maintaining desired decodability properties, such as maximum-distance-separable (MDS) property.
We then introduce a new class of code pairs, which we call \emph{\codenames}, which allow for resource-efficient \conversions.
We begin the study of this new class of code pairs, by focusing on an important regime where \(\Fk = \Cs\Ik\) for any integer \(\Cs \geq 2\) with arbitrary values of \(\In\) and \(\Fn\), which we call the \emph{\regime}.
Furthermore, we focus on  the \emph{\cost} of \codeconversion, which corresponds to \emph{the total number of nodes that participate in the \conversion}.
Keeping the number of nodes accessed small makes \conversion\ less disruptive and allows the unaffected nodes to remain available for serving client requests.
In addition, reducing the number of accesses also reduces disk IO, network bandwidth and CPU consumed.

We prove tight bounds on the \cost\ of \conversions\ for linear MDS codes in the \regime. In particular, our achievability result is an explicit construction of MDS \codenames that are \optimal for all parameters values in the \regime\ albeit with a high field size. 
Finally, we present a sequence of practical low-field-size constructions of \optimal\ \mdscodenames\ in the \regime\ based on Hankel arrays. These constructions lead to a tradeoff between field size and the parameter values they cover with the two extreme points corresponding to (1) \(\Fn - \Fk \leq \lfloor (\In - \Ik)/\Cs \rfloor\) requiring a field size \(q \geq \max\{\In - 1, \Fn - 1\}\), and (2) \(\Fn - \Fk \leq \In - \Ik - \Cs + 1\) requiring a field size \(q \geq \Ik\Ir\).
Thus, our results show that \codeconversions\ can be achieved with a significantly lesser resource overhead as compared to the default approach of re-encoding.
Furthermore, all the constructions presented have the added benefit that they continue to be optimal for a wide range of parameters, which allows to handle the case where the parameters of the final code are unknown a priori.

The rest of the paper is organized as follows.
\Cref{sec:related} discusses related work. 
\Cref{sec:conversion} formalizes the notion of \codeconversions\ and presents a framework for studying \codenames.
\Cref{sec:merge} shows the derivation of lower bounds on the \cost\ of \conversions\ for linear MDS codes in the \regime.
\Cref{sec:construction} describes a general explicit construction for MDS codes in the \regime that meets the \cost\ lower bounds, albeit with a high field size.
\Cref{sec:simple-hankel} describes low-field-size constructions for MDS codes in the \regime, which provide a tradeoff between field size and range of parameter values they cover.
Finally, \cref{sec:conclusion} presents our conclusions and discuss future directions.

\section{Related work}\label{sec:related}

There is extensive literature on the use of erasure codes for reliable data storage. 
In storage systems, failures can be effectively modeled as erasures, and thereby, erasure codes can be used to provide tolerance to failures, at the cost of some storage overhead~\cite{plank05,patterson1988case}.
{Maximum distance separable} (MDS) codes are often used for this purpose, since they achieve the optimal tradeoff between failure tolerance and storage overhead.
A well-known and often-used family of MDS codes is Reed-Solomon codes~\cite{TECC78}. 

When using erasure codes in storage systems, a host of other overheads and performance metrics, in addition to storage overhead, comes into picture. 
Encoding/decoding complexity, node repair performance, degraded read performance, field size, and other metrics can significantly affect real system performance.
Several works in the literature have studied these aspects.

The encoding and decoding of data, and the finite field arithmetic that they require, can be compute intensive.
Motivated by this, \emph{array codes}~\cite{BBBM95,xu1999x,huang2008star,hafner2005weaver} are designed to use XOR operations exclusively, which are typically faster to execute, and aim to decrease the complexity of encoding and decoding.

The repair of failed nodes can incur a large amount of data read and transfer, burdening device IO and network bandwidth.
Several approaches have been proposed to alleviate the impact of repair operations.
Dimakis et al.~\cite{DGWWR10} proposed a new class of codes called \emph{regenerating codes} that minimize the amount of network bandwidth consumed during repair operations.
Several explicit constructions of regenerating codes have been proposed (for example, see \cite{rashmi2011optimal,shah2011interference,suh2011journal,tamo2013zigzag,cadambe2013asymptotic,ye2016explicit,papailiopoulos2013repairTransactions,goparaju2017minimum,chowdhury2018newconstructions,mahdaviani2018bandwidth,mahdaviani2018product}) as well as generalizations (for example, see \cite{shah2010flexible, shum2011cooperative, abdrashitov2017storage}).
It has been shown that meeting the lower bound on the repair bandwidth requirement when MDS property and high rate are desired necessitates a large value for the so called ``sub-packetization'' \cite{tamo2014access,goparaju2014improved,balaji2018tight,alrabiah2019exponential}, which negatively affects certain key performance metrics in storage systems~\cite{rashmi2014hitchhiker}. 
To overcome this issue, several works~\cite{rashmi2013piggybacking,rashmi2017piggybacking,guruswami2017mds} have proposed code constructions that relax the requirement of meeting lower bounds on IO and bandwidth requirements for repair operations.
For example, the Piggybacking framework~\cite{rashmi2017piggybacking} provides a general framework to construct repair-efficient codes by transforming any existing codes, while allowing a small sub-packetization (even as small as \(2\)).
The above discussed works construct vector codes in order to improve the efficiency of repair operation.
The papers~\cite{shanmugam2014repair,guruswami2016repairing,dau2018repairing} propose repair algorithms for (scalar) Reed-Solomon codes that reduce the network bandwidth consumed during repair by downloading elements from a subfield rather than the finite field over which the code is constructed.
Network bandwidth consumed is another metric to optimize for during \conversion. In this paper, we only focus on the access cost.

Another class of codes, called \emph{local codes}~\cite{gopalan2012locality,papailiopoulos2014locally,tamo2014family,kamath2014codes,cadambe2015bounds,tamo2016optimal,tamo2016bounds,barg2017locally,agarwal2018combinatorial,mazumdar2018capacity}, focuses on the locality of codeword symbols during repair, that is, the number of nodes that need to be accessed when repairing a single failure.
Local codes improve repair and degraded read performance, since missing information can be recovered without having to recover the full data.
The locality metric for repair that local codes optimize for is similar to the access cost metric for \conversion\ that we optimize for in this work as both these metrics aim to minimize the number of nodes accessed.
    
There are several classical techniques for creating new codes from existing ones~\cite{TECC78}.
For example, techniques such as \emph{puncturing}, \emph{extending}, \emph{shortening}, and others which can be used to modify codes.
These techniques, however, do not consider the cost of performing such modifications to data that is {already encoded}, which is the focus of our work.

Several works \cite{rashmi2011enabling,MZT18} study the problem of two stage encoding: first generating a certain number of parities during the encoding process and then adding additional parities.
As discussed in~\cite{rashmi2011enabling}, adding additional parities can be conceptually viewed as a repair process by considering the new parity nodes to be generated as failed nodes. Furthermore, as shown in~\cite{shah2011interference}, for MDS codes, the bandwidth requirement for \textit{repair of even a single node} is lower bounded by the same amount as in regenerating codes that require repair of \textit{all} nodes.
Thus one can always employ a regenerating code to add additional parities with minimum bandwidth overhead.
However, when MDS property and high rate are desired, as discussed above, using regenerating codes requires a large sub-packetization.
The paper~\cite{MZT18} employs the Piggybacking framework~\cite{rashmi2013piggybacking,rashmi2017piggybacking} to construct codes that overcome the issue of large sub-packetization factor.
The scenario of adding a fixed number of additional parities, when viewed under the setting of \conversions, corresponds to having \(\Ik = \Fk\) and \(\In < \Fn\).

Another related work \cite{XSBP15} proposes a storage system that uses two erasure codes.
One of the codes prioritizes the network bandwidth required for recovery, while the other prioritizes storage overhead, and data is \converted\ between the two codes according to the workload.
This application constitutes another motivation for resource-efficient conversions.
To reduce the cost of \codeconversion, the system~\cite{XSBP15} uses {product codes}~\cite{TECC78} and {locally repairable codes}~\cite{tamo2014family}, and the local parities are leveraged during \conversion.
The authors, however, choose codes from these two families \emph{ad hoc}, and do not focus on the problem of designing these codes to minimize the cost of \codeconversion.

Several works~\cite{konwar2017layered,cadambe2018ares,wang2017multi} study the update operation in erasure coded storage systems, and the problem of maintaining consistency in such mutable storage systems. 
The cost of updates is another metric to optimize for in \codenames, which we do not consider in this paper.
In the current paper, the focus is on immutable storage systems which comprise a vast majority of large-scale storage systems.

\section{A framework for studying \codeconversions} \label{sec:conversion}

In this section, we formally define and study \emph{\codeconversions} and introduce \emph{\codenames}.

Suppose one wants to \convert\ data that is already encoded using an \([\In, \Ik]\) initial code \(\ICode\) into data encoded using an \([\Fn, \Fk]\) final code \(\FCode\).
Assume, without loss of generality, that each node has a fixed storage capacity \(\alpha\).
In the initial and final configurations, the system stores the same information, but encoded differently.
In order to capture the changes in the dimension of the code during \conversion, we consider \(\NNodes = \Lcm(\Ik, \Fk)\) number of ``message'' symbols (\ie the data to be stored) over a finite field \(\field{q}\), denoted by \(\Msg \in \field{q}^{\NNodes}\).
This corresponds to multiple stripes in the initial and final configurations.
We note that this \emph{need for considering multiple stripes in order to capture the smallest instance of the problem} deviates from existing literature on the repair problem in distributed storage codes where a single stripe is sufficient to capture the problem.

Since there are multiple stripes, we first specify an \emph{initial partition} \(\IPart\) and a \emph{final partition} \(\FPart\) of the set \([\NNodes]\), which map the message symbols of \(\Msg\) to their corresponding initial and final stripes.
The initial partition \(\IPart \subseteq 2^{[\NNodes]}\) is composed of \(\NNodes / \Ik\) disjoint subsets of size \(\Ik\), and the final partition \(\FPart \subseteq 2^{[\NNodes]}\) is composed of \(\NNodes / \Fk\) disjoint subsets of size \(\Fk\).
In the initial (respectively, final) configuration, the data indexed by each subset \(S \in \IPart\ (\text{respectively}, \FPart)\) is encoded using the code \(\ICode\ (\text{respectively}, \FCode)\).
The codewords \(\{\ICode(\Msg_{S}),\, S \in \IPart\}\) are referred to as \emph{initial stripes}, and the codewords \(\{\FCode(\Msg_{S}),\, S \in \FPart\}\) are referred to as \emph{final stripes}, where \(\Msg_{S}\) corresponds to the projection of \(\Msg\) onto the coordinates in \(S\) and \(\Code(\Msg_S)\) is the encoding of \(\Msg_S\) under code \(\Code\).
We now formally define \codeconversion\ and \codenames.

\begin{definition}[\textbf{Code conversion}] \label{def:conversion}
A \emph{\conversion} from an initial code \(\ICode\) to a final code \(\FCode\) with initial partition \(\IPart\) and final partition \(\FPar\) is a procedure, denoted by \(\Transition{\ICode}{\FCode}\), that for any \(\Msg\), takes the set of initial stripes \(\{\ICode(\Msg_S) \mid S \in \IPart\}\) as input, and outputs the corresponding set of final stripes \(\{\FCode(\Msg_S) \mid S \in \FPart\}\).
\end{definition}

The descriptions of the initial and final partitions and codes, along with the \conversion\ procedure, define a \codename.

\begin{definition}[\textbf{\Codename}] \label{def:code}
    A \CodeDefault\ over \(\field{q}\) is defined by: (1) a pair of codes \((\ICode, \FCode)\) where \(\ICode\) is an \([\In, \Ik]\) code over \(\field{q}\) and \(\FCode\) is an \([\Fn, \Fk]\) code over \(\field{q}\); (2) a pair of partitions \(\IPart, \FPart\) of \([\NNodes = \Lcm(\Ik, \Fk)]\) such that each subset in \(\IPart\) is of size \(\Ik\) and each subset in \(\FPart\) is of size \(\Fk\); and (3) a \conversion\ procedure \(\Transition{\ICode}{\FCode}\) that on input \(\{\ICode(\Msg_S) \mid S \in \IPart\}\) outputs \(\{\FCode(\Msg_S) \mid S \in \FPart\}\) for all \(\Msg \in \field{q}^M\). 
\end{definition}

In addition, typically additional constraints on the distance (\ie decodability) of the codes \(\ICode\) and \(\FCode\) would be imposed, such as requiring both codes to be MDS.

\begin{example}
Suppose we want to transition from a \([\In = 3, \Ik = 2]\) code \(\ICode\) to a \([\Fn = 5, \Fk = 3]\) code \(\FCode\).
We consider data \(\Msg\) of length \(M = \Lcm(\Ik = 2, \Fk = 3) = 6\).
In the initial configuration, the data is partitioned into three stripes, each one composed of three \blocks\ encoding two message symbols.
For example, if \(\IPart = \{\{1,2\},\{3,4\},\{5,6\}\}\) then the initial stripes are \(\ICode(\Msg_1,\Msg_2),\, \ICode(\Msg_3,\Msg_4)\), and \(\ICode(\Msg_5,\Msg_6)\).
In the final configuration, the data is partitioned into two stripes, each one composed of five \blocks\ encoding three message symbols.
For example, if \(\FPart = \{\{1,2,3\},\{4,5,6\}\}\) then the final stripes are \(\FCode(\Msg_1,\Msg_2,\Msg_3)\), and \(\FCode(\Msg_4,\Msg_5,\Msg_6)\).
Note that a different valid final partition could have been \(\FPart = \{\{1,3,5\},\{2,4,6\}\}\).

The \conversion\ procedure \(\Transition{\ICode}{\FCode}\) must take \(\{\ICode(\Msg_1,\Msg_2),\,\ICode(\Msg_3,\Msg_4),\,\ICode(\Msg_5,\Msg_6)\}\) as input, and output \(\{\FCode(\Msg_1,\Msg_2,\Msg_3),\allowbreak \FCode(\Msg_4,\Msg_5,\Msg_6)\}\).
In this example, the codes \(\ICode,\, \FCode\), the partitions \(\IPart,\, \FPart\), and procedure \(\Transition{\ICode}{\FCode}\) define a \ParamCode{3}{2}{5}{3}.
\end{example}

\begin{remark}
Note that the definition of \codenames (\cref{def:code}) assumes that \(\ParamsDefault\) are fixed \emph{a priori}, and are known at code construction time. This will be helpful in understanding the fundamental limits of the \conversion process. In practice, this assumption might not always hold. For example, the parameters \(\Fn, \Fk\) depend on the node failure rates that are yet to be observed.
Interestingly, it is indeed possible for a \CodeDefault\ to facilitate \conversion\ for multiple values of \(\Fn, \Fk\), as is the case for the code constructions presented in this paper.
\end{remark}

The overhead of \conversion in a \codename is determined by the cost of the \conversion\ procedure \(\Transition{\ICode}{\FCode}\), as a function of the parameters \(\ParamsDefault\).
Towards minimizing the overhead of the \conversion, our general objective is to design codes \((\ICode, \FCode)\), partitions \((\IPart, \FPart)\) and \conversion\ procedure \(\Transition{\ICode}{\FCode}\) that satisfy \cref{def:code} and minimize the \conversion\ cost for given parameters \(\ParamsDefault\), subject to desired decodability constraints on \(\ICode\) and \(\FCode\).

Depending on the relative importance of various resources in the cluster, one might be interested in optimizing the \conversion\ with respect to various types of costs such as access, network bandwidth, disk IO, CPU, etc., or a combination of these costs.
The general formulation of \codeconversions\ above provides a powerful framework to theoretically reason about \codenames.
In what follows, we will focus on a specific regime and a specific cost model.

\section{Lower bounds on \cost\ of \codeconversion} \label{sec:merge}

The focus of this section is on deriving lower bounds on the \cost of \codeconversion.
We consider one of the fundamental regimes of \codenames, that corresponds to merging several initial stripes of a code into a single, longer final stripe.
Specifically, the \codenames\ in this regime have \(\Fk = \Cs\Ik\), where \(\Cs \geq 2\) is the number of initial stripes merged, with arbitrary values of \(\In\) and \(\Fn\).
We call this regime as \textit{\regime}.
We additionally require that both the initial and final code are linear and MDS.
Since linear MDS codes are widely used in storage systems and are well understood in the Coding Theory literature, they constitute a good starting point.

We focus on the \textit{access cost} of \codeconversion, that is, the number of \blocks\ that are affected by the \conversion.
The \cost\ of \conversion\ measures the total number of \blocks\ \emph{accessed} during \conversion.
Each new \block\ needs to be {written}, and hence requires accessing a node.
Similarly, each \block\ from the initial stripes that is {read}, requires accessing a node.
Therefore, minimizing \cost\ amounts to \textit{minimizing the sum of the number of new \blocks\ written and the number of \blocks\ read from the initial stripes}.\footnote{Readers who are familiar with the literature on regenerating codes might observe that \codenames optimizing for the \cost are ``scalar'' codes as opposed to being ``vector'' codes.}
Keeping this number small makes \codeconversion\ less disruptive and allows the unaffected nodes to remain available for application-specific purposes throughout the procedure, for example, to serve client requests in a storage system.
Furthermore, reducing the number of accesses also reduces disk IO, network bandwidth and CPU consumed.

In \cref{sec:construction}, we will show that the lower bounds on the \cost derived in this section are in fact achievable.
Therefore, we refer to \mdscodenames\ in the \regime\ that achieve these lower bounds as \emph{\optimal}.

\begin{definition}[Access-optimal]
    A linear MDS \RegimeCodeDefault\ is said to be \emph{\optimal} if and only if it attains the minimum \cost\ over all linear MDS \RegimeCodesDefault.
\end{definition}

We first start with a description of the notation in \cref{sec:notation} and then derive lower bounds on the \cost in \cref{sec:lower-bounds}.

\subsection{Notation} \label{sec:notation} 

Let \(\ICode\) be an \([\In, \Ik]\) MDS code over field \(\field{q}\), specified by generator matrix \(\IGen\), with columns (that is, \encvecs) \(\{\ICol{1}, \ldots, \ICol{\In}\} \subseteq \field{q}^{\Ik}\).
Let \(\Cs \geq 2\) be an integer, and let \(\FCode\) be an \([\Fn, \Fk = \Cs\Ik]\) MDS code over field \(\field{q}\), specified by generator matrix \(\FGen\), with columns (that is, \encvecs) \(\{\FCol{1}, \ldots, \FCol{\Fn}\} \subseteq \field{q}^{\Fk}\).
Let \(\Ir = \In - \Ik\) and \(\Fr = \Fn - \Fk\).
When \(\ICode\) and \(\FCode\) are systematic, \(\Ir\) and \(\Fr\) correspond to the initial number of parities and final number of parities, respectively.
All vectors are assumed to be column vectors. We will use the notation \(\V{v}[l]\) to denote the \(l\)-th coordinate of a vector \(\V{v}\).

We will represent \emph{all} the code symbols in the initial stripes as being generated by a single \(\Cs\Ik \times \Cs\In\) matrix \(\IGenp\), with \encvecs\ \(\{\IColp{i}{j} \mid i \in [\Cs], j \in [\In]\} \subseteq \field{q}^{\Fk}\).
This representation can be viewed as embedding the column vectors of the generator matrix \(\IGen\) in an \(\Cs\Ik\)-dimensional space, where the index set \(\StripeIndexSet{i} = \{(i - 1)\Ik + 1, \ldots, i\Ik\}, i \in [\Cs]\) corresponds to the \encvecs\ for initial stripe \(i\).
Let \(\IColp{i}{j}\) denote the \(j\)-th encoding vector in the initial stripe \(i\) in this (embedded) representation. Thus, \(\IColp{i}{j}[l] = \ICol{j}[l - (i - 1)\Ik]\) for \(l \in \StripeIndexSet{i}\), and \(\IColp{i}{j}[l] = 0\) otherwise.
As an example, \cref{fig:vector-notation-ex} shows the values of the defined terms for the single parity-check code from \cref{fig:combine_ex} with \(\In = 3, \Ik = 2, \Fn = 5, \Fk = 4\).

\begin{figure}
    \centering
    \includegraphics[width=.8\textwidth]{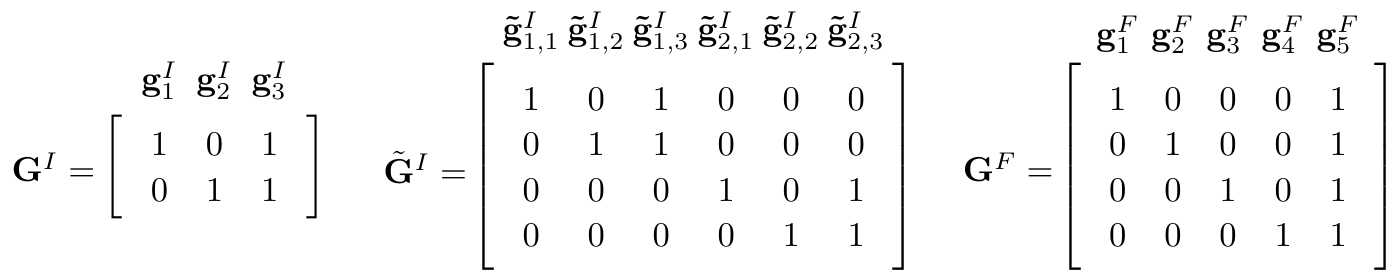}
    \caption{Generator matrices for a specific \ParamCode{3}{2}{5}{4}: \(\IGen\) is the generator matrix of the initial code; \(\IGenp\) is the generator matrix of all initial stripes; \(\FGen\) is the generator matrix of the final code.}
    \label{fig:vector-notation-ex}
\end{figure}

At times, focus will be only on the coordinates of an \encvec\ of a certain initial stripe \(i\).
For this purpose, define \(\Proj_{\StripeIndexSet{i}}(\V{v}) \in \field{q}^{\Ik}\) to be the projection of \(\V{v} \in \field{q}^{\Fk}\) to the coordinates in an index set \(\StripeIndexSet{i}\), and for a set \(\Vectors\) of vectors, \(\Proj_{\StripeIndexSet{i}}(\Vectors) = \{\Proj_{\StripeIndexSet{i}}(\V{v}) \mid \V{v} \in \Vectors\}\).
For example, \(\Proj_{\StripeIndexSet{i}}(\IColp{i}{j}) = \ICol{j}\) for all \(i \in [\Cs]\) and \(j \in [\In]\).

The following sets of vectors are defined: the \encvecs\ from initial stripe \(i\), \(\StripeICols{i} = \{\IColp{i}{j} \mid j \in [\In]\}\), all the \encvecs\ from all the initial stripes, \(\AllICols = \cup_{i \in [\Cs]} \StripeICols{i}\), and all the \encvecs\ from the final stripe \(\AllFCols = \{\FCol{j} \mid j \in [\Fn]\}\).

We use the term \emph{unchanged \blocks} to refer to \blocks\ from the initial stripes that remain as is (that is, unchanged) in the final stripe.
The \blocks\ in the final stripe that were not present in the initial stripes are called \emph{new}, and the \blocks\ from the initial stripes that do not carry over to the final stripe are called \emph{retired}.
For example, in \cref{fig:combine_ex}, all the data \blocks\ are unchanged \blocks\ (unshaded boxes), the single parity \block\ of the final stripe is a new \block, and the two parity blocks from the initial stripes are retired \blocks.
Each unchanged \block\ corresponds to a pair of identical initial and final \encvecs, that is, a tuple of indices \((i, j, l)\) such that \(\IColp{i}{j} = \FCol{l}\).
For instance, the example in \cref{fig:combine_ex} has four unchanged \blocks, corresponding to the identical \encvecs\ \(\IColp{i}{j} = \FCol{2(i - 1) + j}\) for \(i, j \in [2]\).
The final \encvecs\ \(\AllFCols\) can thus be partitioned into the following sets: \emph{unchanged \encvecs} from initial stripe \(i\), \(\UnchangedStripeCols{i} = \AllFCols \cap \StripeICols{i}\) for all \(i \in [\Cs]\), and \emph{new \encvecs} \(\NewCols = \AllFCols \setminus \AllICols\).

From the point of view of \conversion\ cost, unchanged \blocks\ are ideal, because they require no extra work.
On the other hand, constructing new \blocks\ require accessing \blocks\ from the initial stripes.
When a \block\ from the initial stripes is accessed, all of its contents are downloaded to a central location, where they are available for the construction of all new \blocks.
For example, in \cref{fig:combine_ex}, one \block\ from each initial stripe is accessed during \conversion.

During \conversion, new \blocks\ are constructed by reading \blocks\ from the initial stripes.
That is, every new \encvec\ is simply a linear combination of a specific subset of \(\AllICols\).
Define the \emph{read access set} for an MDS \RegimeCodeDefault\ as the set of tuples \(\Down \in [\Cs] \times [\In]\) such that the set of new \encvecs \(\NewCols\) is contained in the span of the set \(\{\IColp{i}{j} \mid (i, j) \in \Down\}\).
Furthermore, define the index sets \(\Down_i = \{j \mid (i, j) \in \Down\}\), \(\forall i \in [\Cs]\) which denote the \encvecs\ accessed from each initial stripe.

\subsection{Lower bounds on the \cost\ of \codeconversion} \label{sec:lower-bounds}

In this subsection, we present lower bounds on the \cost\ of linear \mdscodenames\ in the \regime.
This is done in four steps:
\begin{enumerate}
    \item We show that in the \regime, all possible pairs of partitions \(\IPar\) and \(\FPar\) partitions are equivalent up to relabeling, and hence do not need to be specified.
    \item An upper bound on the maximum number of unchanged \blocks\ is proved. We call \codenames\ that meet this bound as ``\emph{stable}''.
    \item Lower bounds on the \cost\ of linear \mdscodenames\ are proved, under the added restriction that the \codenames\ are stable.
    \item The stability restriction is removed, by showing that non-stable linear \mdscodenames\ necessarily incur higher \cost, and hence it suffices to consider only stable \mdscodenames.
\end{enumerate}

We now start with the first step. 
In the general regime, partition functions need to be specified since they indicate how message symbols from the initial stripes are mapped into the final stripes.
In the \regime, however, there is only one final stripe, and hence the choice of the partition functions does not matter. 

\begin{prop}
    For every \RegimeCodeDefault, all possible pairs of initial and final partitions \((\IPart, \FPart)\) are equivalent up to relabeling.
\end{prop}
\begin{proof}
    Given that \(M = \Lcm(\Ik, \Cs\Ik) = \Cs\Ik\), there is only one possible final partition \(\FPart = \{[\Cs\Ik]\}\).
    Thus, regardless of \(\IPart\), all data in the initial stripes will get mapped to the same final stripe.
    By relabeling blocks, any two initial partitions can be made equivalent.
\end{proof}

Thus, the analysis of \codenames\ in the \regime\ in this regime can be simplified by noting that the choice of partitions \(\IPar\) and \(\FPar\) is inconsequential.

Since one of the terms in \cost\ is the number of new \blocks, a natural way to reduce \cost\ is to maximize the number of unchanged \blocks.
However, there is a limit on the number of \blocks\ that can remain unchanged.
\begin{prop}\label{thm:max-unchanged}
    In an MDS \RegimeCodeDefault, there can be at most \(\Ik\) unchanged vectors from each initial stripe.
    Thus, there can be at most \(\Cs\Ik\) unchanged vectors in total, or in other words, there will be at least \(\Fr\) new vectors.
\end{prop}
\begin{proof}
    Every subset \(\Vectors \subseteq \StripeICols{i}\) of size at least \(\Ik + 1\) is linearly dependent, and thus if \(\Vectors \subseteq \AllFCols\) then \(\FCode\) cannot be MDS.
    Hence, for each stripe \(i \in [\Cs]\), the amount of unchanged vectors \(|\UnchangedStripeCols{i}|\) is at most \(\Ik\).
\end{proof}
Since new \blocks\ are constructed using only the contents of \blocks\ read, it is clear that both the quantities that make up \cost\ are going to be related.
Intuitively, more new \blocks\ means that more \blocks\ need to be read, resulting in higher \cost.
With this intuition in mind, we will first focus on the case where the number of new \blocks\ is the minimum: \(|\NewCols| = \Fn - \Cs\Ik = \Fn - \Fk = \Fr\).
We refer to such codes as \textit{stable} \codenames.

\begin{definition}[Stability]\label{def:stable}
    An MDS \RegimeCodeDefault\ is \emph{stable} if and only if it has exactly \(\Cs\Ik\) unchanged \blocks, or in other words, exactly \(\Fr\) new blocks.
\end{definition}
We first prove lower bounds on the \cost of stable linear MDS \codenames, and then show that the \cost of \conversion\ in MDS codes without this stability property can only be higher.

A natural question now is characterizing the minimum size of the read access set for \conversion \(\Down\) for MDS codes.
Clearly, accessing \(\Ik\) \blocks\ from each initial stripe will always suffice, since this is sufficient to decode all the original data.
Thus, in a minimum size \(\Down\) we can upper bound the size of each \(\Down_i\) by \(|\Down_i| \leq \Ik,\; i \in [\Cs]\).

The first lower bound on the size of \(\Down_i\) will be given by the interaction between \(\Fr\) and the MDS property.
\begin{lemma}\label{thm:down-size:one}
    For all linear stable MDS \RegimeCodesDefault, the read access set \(\Down_i\) from each initial stripe \(i \in [\Cs]\) satisfies \(|\Down_i| \geq \min\{\Ik, \Fr\}\).
\end{lemma}

\begin{proof}
    By the MDS property, every subset \(\Vectors \in \AllFCols\) of size at most \(\Fk = \Cs\Ik\) is linearly independent.
    For any initial stripe \(i \in [\Cs]\), consider the set of all unchanged \encvecs\ from other stripes, \(\cup_{\ell \neq i} \StripeICols{\ell}\), and pick any subset of new \encvecs\ \(\Vectorss \subseteq \NewCols\) of size \(|\Vectorss| = \min\{\Ik, \Fr\}\).
    Consider the subset \(\Vectors = (\cup_{\ell \neq i} \StripeICols{\ell} \cup \Vectorss)\): it is true that \(\Vectors \subseteq \AllFCols\) and \(|\Vectors| = (\Cs - 1)\Ik + \min\{\Ik, \Fr\} \leq \Fk\).
    Therefore, all the \encvecs\ in \(\Vectors\) are linearly independent.
    
    Notice that the \encvecs\ in \(\Vectors \setminus \Vectorss\) contain no information about initial stripe \(i\) and complete information about every other initial stripe \(\ell \neq i\).
    Therefore, the information about initial stripe \(i\) in each \encvec\ in \(\Vectorss\) has to be linearly independent since, otherwise, \(\Vectors\) could not be linearly independent.
    Formally, it must be the case that \(\Vectorss_i = \Proj_{\StripeIndexSet{i}}(\Vectorss)\) has rank equal to \(\min\{\Ik, \Fr\}\) (recall from \cref{sec:notation} that \(\StripeIndexSet{i}\) is the set of coordinates belonging to initial stripe \(i\)).
    However, by definition, the subset \(\Vectorss_i\) must be contained in the span of \(\{\ICol{j} \mid j \in \Down_i\}\).
    Therefore, the rank of \(\{\ICol{j} \mid j \in \Down_i\}\) is at least that of \(\Vectorss_i\), which implies that \(|\Down_i| \geq \min\{\Ik, \Fr\}\).
\end{proof}
Therefore, in general we need to access at least \(\Fr\) vectors from each initial stripe, unless \(\Fr \geq \Ik\), in which case we need to access \(\Ik\) \encvecs, that is, the full data.

We next show that, in a linear MDS stable \codename\ in the \regime, when the number of new \blocks\ \(\Fr\) is bigger than \(\Ir\), at least \(\Ik\) \blocks\ need to be accessed from each initial stripe.
The intuition behind this result is the following: in an MDS stable \codename\ in the \regime, when the number of new \blocks\ \(\Fr\) is bigger than \(\Ir\), during a \conversion one is forced to read more than \(\Ir\) \blocks.
Hence there must exist \blocks\ from the initial stripes that are both unchanged and are read during \conversion.
Since the unchanged blocks that are read are also present in the final stripe, the information read from these \blocks\ is not useful in creating a new \block\ that retains the MDS property for the final code  unless \(\Ik\) \blocks (that is, full data) are read.

\begin{lemma}\label{thm:down-size:two}
    For all linear stable MDS \RegimeCodesDefault, if \(\Ir < \Fr\) then the read access set \(\Down_i\) from each initial stripe \(i \in [\Cs]\) satisfies \(|\Down_i| \geq \Ik\).
\end{lemma}

\begin{proof}
    When \(\Fr \geq \Ik\), this lemma is equivalent to \cref{thm:down-size:one}, so assume \(\Ir < \Fr < \Ik\).
    From the proof of \cref{thm:down-size:one}, for every initial stripe \(i \in [\Cs]\) it holds that \(|\Down_i| \geq \Fr\).
    Since \(\Fr > \Ir\), this implies that \(\Down_i\) must contain at least one index of an unchanged \encvec.
    
    Choose a subset of at most \(\Fk = \Cs\Ik\) \encvecs\ from \(\AllFCols\), which must be linearly independent by the MDS property.
    In this subset, include all the unchanged \encvecs\ from the other initial stripes, \(\cup_{l \neq i} \StripeICols{l}\).
    Then, choose all the unchanged \encvecs\ from initial stripe \(i\) that are accessed during \conversion, \(\Vectorss_1 = (\{\IColp{i}{j} \mid j \in \Down_i\} \cap \UnchangedStripeCols{i})\).
    For the remaining vectors (if any), choose an arbitrary subset of new \encvecs, \(\Vectorss_2 \subseteq \NewCols\), such that:
    \begin{equation}\label{eq:set-size}
        |\Vectorss_2| = \min\{\Ik - |\Vectorss_1|, \Fr\}.
    \end{equation}
    It is easy to check that the subset \(\Vectors = \cup_{l \neq i} \StripeICols{l} \cup \Vectorss_1 \cup \Vectorss_2\) is of size at most \(\Fk = \Cs\Ik\), and therefore it is linearly independent.
    This choice of \(\Vectors\) follows from the idea that the information contributed by \(\Vectorss_1\) to the new \encvecs\ is already present in the unchanged \encvecs, which will be at odds with the linear independence of \(\Vectors\).
    
    Since the elements of \(\Vectorss_1\) and \(\Vectorss_2\) are the only \encvecs\ in \(\Vectors\) that contain information from initial stripe \(i\), it must be the case that \(\widetilde{\Vectorss} = \Proj_{\StripeIndexSet{i}}(\Vectorss_1) \cup \Proj_{\StripeIndexSet{i}}(\Vectorss_2)\) has rank \(|\Vectorss_1| + |\Vectorss_2|\).
    Moreover, \(\widetilde{\Vectorss}\) is contained in the span of \(\{\ICol{j} \mid j \in \Down_i\}\) by definition, so it holds that: 
    \begin{equation}\label{eq:down-lower}
      |\Down_i| \geq |\Vectorss_1| + |\Vectorss_2|.
    \end{equation}
    
    From \cref{eq:set-size}, there are two cases:
    
    \noindent\textbf{Case 1:} \(\Ik - |\Vectorss_1| \leq \Fr\).
    Then \(|\Vectorss_2| = \Ik - |\Vectorss_1|\) and by \cref{eq:down-lower} it holds that \(|\Down_i| \geq |\Vectorss_1| + |\Vectorss_2| = \Ik\).
    
    \noindent\textbf{Case 2:} \(\Ik - |\Vectorss_1| > \Fr\).
    Then \(|\Vectorss_2| = \Fr\) and by \cref{eq:down-lower} it holds that:
    \begin{equation}\label{eq:first-ineq}
        |\Down_i| \geq |\Vectorss_1| + \Fr.
    \end{equation}
    Notice that there are only \(\Ir\) retired (i.e.\ not unchanged) \encvecs\ in stripe \(i\).
    Since every accessed \encvec\ is either in \(\Vectorss_1\) or is a retired \encvec, it holds that:
    \begin{equation}\label{eq:second-ineq}
        |\Down_i| \leq |\Vectorss_1| + \Ir. 
    \end{equation}
    By combining \cref{eq:first-ineq} and \cref{eq:second-ineq}, we arrive at the contradiction \(\Fr \leq \Ir\), which occurs because there are not enough retired \blocks\ in the initial stripe \(i\) to ensure that the final code has the MDS property.
    Therefore, case 1 always holds, and \(|\Down_i| \geq k\).
\end{proof}

Combining the above results leads to the following theorem on the lower bound of read access set size of linear stable MDS convertible codes.
\begin{theorem}\label{thm:min-down-size}
    Let \(\MinDownSizeDefault\) denote the minimum integer \(\DownSize\) such that there exists a linear stable MDS \RegimeCodeDefault\ with read access set \(\Down\) of size \(|\Down| = d\).
    For all valid parameters, \(\MinDownSizeDefault \geq \Cs \min\{\Ik, \Fr\}\).
    Furthermore, if \(\Ir < \Fr\), then \(\MinDownSizeDefault \geq \Cs\Ik\).
\end{theorem}
\begin{proof}
    Follows directly from \cref{thm:down-size:one} and \cref{thm:down-size:two}.
\end{proof}

So far we have focused on deriving lower bounds on the \cost of \conversion\ for \textit{stable} MDS \codenames, which have the maximum number of unchanged \blocks.
That is, \codenames that have \(\Cs\Ik\) unchanged \blocks\ and \(\Fr\) new \blocks.
We next show that this lower bound generally applies even for non-stable \codenames by proving that increasing the number of new \blocks\ from the minimum possible does not decrease the lower bound on the size of the read access set \(\Down\).

\begin{lemma} \label{thm:min-down-size-non-stable}
    The lower bounds on the size of the read access set from  \Cref{thm:min-down-size} hold for all (including non-stable) linear MDS \RegimeCodesDefault.
\end{lemma}

\begin{proof}
    We show that, even for non-stable \codenames, that is, when there are more than \(\Fr\) new \blocks, the bounds on the read access set \(\Down\) from \cref{thm:min-down-size} still hold.
    
    \noindent\textbf{Case 1:} \(\Ir \geq \Fr\).
    Let \(i \in [\Cs]\) be an arbitrary initial stripe.
    We lower bound the size of \(\Down_i\) by invoking the MDS property on a subset \(\Vectors \subseteq \AllFCols\) of size \(|\Vectors| = \Cs\Ik\) that minimizes the size of the intersection \(|\Vectors \cap \UnchangedStripeCols{i}|\). 
    There are exactly \(\Fr\) \encvecs\ in \(\AllFCols \setminus \Vectors\), so the minimum size of the intersection \(|\Vectors \cap \UnchangedStripeCols{i}|\) is \(\max\{|\UnchangedStripeCols{i}| - \Fr, 0\}\).
    Clearly, the subset \(\Proj_{\StripeIndexSet{i}}(\Vectors)\) has rank \(\Ik\) due to the MDS property.
    Therefore, it holds that \(|\Down_i| + \max\{|\UnchangedStripeCols{i}| - \Fr, 0\} \geq \Ik\).
    By reordering, the following is obtained:
    \begin{equation*}\label{eq:stripe-reuse}
        |\Down_i| \geq \Ik - \max\{|\UnchangedStripeCols{i}| - \Fr, 0\} \geq \min\{\Fr, \Ik\},
    \end{equation*}
    which means that the bound on \(\Down_i\) established in \cref{thm:down-size:one} continues to hold for non-stable codes.
    
    \noindent\textbf{Case 2:} \(\Ir < \Fr\).
    Let \(i \in [\Cs]\) be an arbitrary initial stripe, let \(\Vectorss_1 = (\{\IColp{i}{j} \mid j \in \Down_i\} \cap \UnchangedStripeCols{i})\) be the unchanged \encvecs\ that are accessed during \conversion, and let \(\Vectorss_2 = \UnchangedStripeCols{i} \setminus \Vectorss_1\) be the unchanged \encvecs\ that are \emph{not} accessed during \conversion.
    Consider the subset \(\Vectors \subseteq \AllFCols\) of \(\Fk = \Cs\Ik\) \encvecs\ from the final stripe such that \(\Vectorss_1 \subseteq \Vectors\) and the size of the intersection \(\Vectorss_3 = (S \cap \Vectorss_2)\) is minimized.
    Since \(\Vectors\) may exclude at most \(\Fr\) \encvecs\ from the final stripe, it holds that:
    \begin{equation}\label{eq:set-size-two}
        |\Vectorss_3| = \max\{0, |\Vectorss_2| - \Fr\}.
    \end{equation}
    
    By the MDS property, \(\Vectors\) is a linearly independent set of \encvecs\ of size \(\Fk\), and thus, must contain all the information to recover the contents of every initial stripe, and in particular, initial stripe \(i\).
    Since all the information in \(\Vectors\) about stripe \(i\) is in either \(\Vectorss_3\) or the accessed \encvecs, it must hold that:
    \begin{equation} \label{eq:down-lower-two}
        |\Down_i| + |\Vectorss_3| \geq \Ik.
    \end{equation}
    
    From \cref{eq:set-size-two}, there are two cases:
    
    \noindent\textbf{Subcase 2.1:} \(|\Vectorss_2| - \Fr \leq 0\).
    Then \(|\Vectorss_3| = 0\), and by \cref{eq:down-lower-two} it holds that \(|\Down_i| \geq \Ik\), which matches the bound of \cref{thm:down-size:two}.
    
    \noindent\textbf{Subcase 2.2:} \(|\Vectorss_2| - \Fr > 0\).
    Then \(|\Vectorss_3| = |\Vectorss_2| - \Fr\), and by \cref{eq:down-lower-two} it holds that:
    \begin{equation}\label{eq:first-ineq-two}
        |\Down_i| + |\Vectorss_2| - \Fr \geq \Ik.
    \end{equation}
    The initial stripe \(i\) has \(\Ik + \Ir\) \blocks.
    By the principle of inclusion-exclusion we have that:
    \begin{equation}\label{eq:second-ineq-two}
        |\Down_i| + |\UnchangedStripeCols{i}| - |\Vectorss_1| \leq \Ik + \Ir.
    \end{equation}
    By using \cref{eq:first-ineq-two}, \cref{eq:second-ineq-two} and the fact that \(|\Vectorss_2| = |\UnchangedStripeCols{i}| - |\Vectorss_1|\), we conclude that \(\Ir \geq \Fr\), which is a contradiction and means that subcase 2.1 always holds in this case.
\end{proof}

The above result, along with the fact that the lower bound in \cref{thm:min-down-size} is achievable (as will be shown in \cref{sec:construction}), implies that all \optimal\ linear MDS \codenames in the \regime\ have the minimum possible number of new \blocks\ (which is \(\Fr\) as shown in \cref{thm:max-unchanged}), that is they are stable.

\begin{lemma} \label{thm:optimal-stable}
    All \optimal\ MDS \RegimeCodesDefault\ are stable.
\end{lemma}

\begin{proof}
    \Cref{thm:min-down-size-non-stable} shows that the lower bound on the read access set \(\Down\) for stable linear MDS \codenames continues to hold in the non-stable case.
    Furthermore, this bound is achievable by stable linear MDS \codenames\ in the \regime (as will be shown in \cref{sec:construction}).
    The number of new blocks written during \conversion\ under stable MDS \codenames is \(\Fr\).
    On the other hand, the number of new \blocks\ under a non-stable \codename\ is strictly greater than \(\Fr\).
    Thus, the overall \cost\ of a non-stable MDS \RegimeCodeDefault\ is strictly greater than the \cost\ of an \optimal\ \RegimeCodeDefault.
\end{proof}
Thus, for MDS \codenames\ in the \regime, it suffices to focus only on stable codes.
Combining all the results above, leads to the following key result.
\begin{theorem}\label{thm:min-access-cost-final}
For all linear MDS \RegimeCodesDefault, the \cost\ of \conversion\ is at least \(\Fr + \Cs\min\{\Ik,\Fr\}\).
    Furthermore, if \(\Ir < \Fr\), the \cost\ of \conversion\ is at least \(\Fr + \Cs\Ik\).
\end{theorem} 
\begin{proof}
    Follows from \cref{thm:min-down-size}, \cref{thm:min-down-size-non-stable}, and the definition of \cost.
\end{proof}
In \cref{sec:construction} we show that the lower bound of  \cref{thm:min-access-cost-final} is achievable for all parameters. Thus, \Cref{thm:min-access-cost-final} implies that it is possible to perform \conversion\ of MDS \codenames\ in the \regime\ with significantly less \cost\ than the na\"ive strategy if and only if \(\Fr \leq \Ir < \Ik\).
For example, for an MDS \ParamCode{14}{10}{24}{20} the na\"ive strategy has an \cost\ of \(\Fn = 24\), while the optimal \cost\ is \((\Cs + 1)\Fr = 12\), which corresponds to savings in \cost\ of \(50\%\).

\section{Achievability: Explicit \optimal\ \codenames in the \regime}\label{sec:construction}

In this section, we present an explicit construction of \optimal\ MDS \codenames\ for all parameters in the \regime.
In \cref{sec:general-details}, we describe the construction of the generator matrices for the initial and final code.
Then, in \cref{sec:general-proof}, we prove that the code described by this construction has optimal \cost during \codeconversion.

\subsection{Explicit construction} \label{sec:general-details}
Recall that, in the \regime, \(\Fk = \Cs\Ik\), for any integer \(\Cs \geq 2\) and arbitrary \(\In\) and \(\Fn\). Also, recall that \(\Ir = \In - \Ik\) and \(\Fr = \Fn - \Fk\).
Notice that when \(\Ir < \Fr\), or \(\Ik \leq \Fr\), constructing an \optimal\ \codename\ is trivial.
In those cases, one can simply access all the \(\Fk = \Cs\Ik\) data blocks of the initial stripes, which meets the bound stated in \cref{thm:min-down-size}.
Thus, assume \(\Fr \leq \min\{\Ir, \Ik\}\).

Let \(\IGen, \FGen\) be the generator matrices of \(\ICode, \FCode\) respectively.
Our construction is systematic, that is, both \(\ICode\) and \(\FCode\) are systematic MDS codes.
Thus \(\IGen, \FGen\) are of the form \(\IGen = [\IdMat | \IPar]\) and \(\FGen = [\IdMat | \FPar]\), where \(\IPar\) is a \(\Ik \times \Ir\) matrix and \(\FPar\) is a \(\Fk \times \Fr\) matrix.
Therefore, to define the initial and final code, only \(\IPar\) and \(\FPar\) need to be specified. 
Let \(\field{q}\) be a finite field of size \(q = p^{D}\), where \(p\) is any prime and the degree \(D\) depends on the \codename\ parameters and will be specified later in this section.
Let \(\theta\) be a primitive element of \(\field{q}\).

Define entry \((i, j)\) of \(\IPar \in \field{q}^{\Ik \times \Ir}\) as \(\theta^{(i - 1)(j - 1)}\), where \((i, j)\) ranges over \([\Ik] \times [\Ir]\).
Entry \((i, j)\) of \(\FPar \in \field{q}^{\Fk \times \Fr}\) is defined in an identical fashion, as \(\theta^{(i - 1)(j - 1)}\), where \((i, j)\) ranges over \([\Fk] \times [\Ir]\).

For example, for \(\Ik = 3, \Ir = 3, \Fk = 6, \Fr = 3\), the matrices \(\IPar\) and \(\FPar\) would be:
\begin{equation*}
   \IPar =
   \begin{bmatrix}
     1 & 1 & 1 \\
     1 & \theta & \theta^2 \\
     1 & \theta^2 & \theta^4 \\
   \end{bmatrix}
   \qquad
   \FPar =
   \begin{bmatrix}
     1 & 1 & 1 \\
     1 & \theta & \theta^2 \\
     1 & \theta^2 & \theta^4 \\
     1 & \theta^3 & \theta^6 \\
     1 & \theta^4 & \theta^8 \\
     1 & \theta^5 & \theta^{10} \\
   \end{bmatrix}
\end{equation*}

Our explicit construction is \textit{stable} (recall from \cref{thm:optimal-stable} that all \optimal\ MDS \codenames\ in the \regime\ are stable), that is, it has exactly \(\Fk = \Cs\Ik\) unchanged \encvecs.
Given that our construction is also systematic it follows that these unchanged \encvecs\ correspond exactly to the systematic elements of \(\FCode\).

\subsection{Proof of optimal \cost during \conversion} \label{sec:general-proof}

Throughout this section, we use the following notation for submatrices: let \(M\) be a \(n \times m\) matrix, the submatrix of \(M\) defined by row indices \(\{i_1, \ldots, i_a\}\) and column indices \(\{j_1, \ldots, j_b\}\) is denoted by \(M[i_1, \ldots, i_a; j_1, \ldots, j_b]\).
For conciseness, we use \(*\) to denote all row or column indices, e.g., \(M[*;j_1,\ldots,j_b]\) denotes the submatrix composed by columns \(\{j_1,\ldots,j_b\}\), and \(M[i_1,\ldots,i_a;*]\) denotes the submatrix composed by rows \(\{i_1,\ldots,i_a\}\).

We first recall an important fact about systematic MDS codes.
\begin{prop}[\cite{TECC78}]\label{thm:superregular}
    Let \(\Code\) be an \([n, k]\) code with generator matrix \(G = [I | P]\). 
    Then \(\Code\) is MDS if and only if \(P\) is superregular, that is, every square submatrix of \(P\) is nonsingular\footnote{This definition of superregularity is different from the definition introduced in~\cite{GRS06}, which is sometimes used in the context of convolutional codes.}.
    \hfill\qed
\end{prop}
Thus, to be MDS, both \(\IPar\) and \(\FPar\) need to be superregular.

From the bound in \cref{thm:down-size:one}, to be \optimal during \conversion\ when \(\Fr \leq \Ik\), the columns of \(\FPar\) (that is, the new \encvecs) have to be such that they can be constructed by only accessing \(\Fr\) columns of \(\IGen\) (that is, the initial \encvecs) during \conversion. Thus, it suffices to show that the columns of \(\FPar\) can be constructed by accessing only \(\Fr\) columns of \(\IPar\) during \conversion.
To capture this property, we introduce the following definition.
\begin{definition}[\textbf{\Constructible{t}}]\label{def:constructible}
    We will say that an \(n \times m_1\) matrix \(M_1\) is \emph{\Constructible{t}} from an \(n \times m_2\) matrix \(M_2\) if and only if there exists a subset \(S \subseteq \Cols(M_2)\) of size \(t\), such that the \(m_1\) columns of \(M_1\) are in the span of \(S\).
    We say that a \(\lambda n \times m_1\) matrix \(M_1\) is \emph{\BlockConstructible{t}} from an \(n \times m_2\) matrix \(M_2\) if and only if for every \(i \in [\Cs]\), the submatrix \(M_1[(i - 1)n + 1, \ldots, in; *]\) is \Constructible{t} from \(M_2\).
\end{definition}

\begin{theorem}\label{thm:construction-conditions}
    A systematic \RegimeCodeDefault\ with \(\Ik \times \Ir\) initial parity generator matrix \(\IPar\) and \(\Fk \times \Fr\) final parity generator matrix \(\FPar\) is MDS and \optimal, if the following two conditions hold: (1) if \(\Ir \geq \Fr\) then \(\FPar\) is \BlockConstructible{\Fr} from \(\IPar\), and (2) \(\IPar, \FPar\) are superregular.
\end{theorem}
\begin{proof}
    Follows from \cref{thm:superregular} and \cref{def:constructible}.
\end{proof}

Thus, we can reduce the problem of proving the optimality of a systematic MDS \codename\ in the \regime\ to that of showing that matrices \(\IPar\) and \(\FPar\) satisfy the two properties mentioned in \cref{thm:construction-conditions}.

We first show that the construction specified in \cref{sec:general-details} satisfies condition (1) of \cref{thm:construction-conditions}.
\begin{lemma}\label{thm:condition-one}
    Let \(\IPar, \FPar\) be as defined in \cref{sec:general-details}.
    Then \(\FPar\) is \BlockConstructible{\Fr} from \(\IPar\).
\end{lemma}

\begin{proof}
Consider the first \(\Fr\) columns of \(\IPar\), which we denote as \(\IPar_{\Fr} = \IPar[*;1,\ldots,\Fr]\).
Notice that \(\FPar\) can be written as the following block matrix:
\begin{equation*}
    \FPar =
    \begin{bmatrix}
        \IPar_{\Fr} \\
        \IPar_{\Fr} \Diag(1, \theta^{\Ik}, \theta^{2\Ik}, \ldots, \theta^{\Ik(\Fr - 1)}) \\
        \IPar_{\Fr} \Diag(1, \theta^{2\Ik}, \theta^{2\cdot2\Ik}, \ldots, \theta^{2\Ik(\Fr - 1)}) \\
        \vdots \\
        \IPar_{\Fr} \Diag(1, \theta^{(\Cs - 1)\Ik}, \theta^{2(\Cs - 1)\Ik}, \ldots, \theta^{(\Cs - 1)\Ik(\Fr - 1)}) \\
    \end{bmatrix},
\end{equation*}
where \(\Diag(a_1, a_2, \ldots, a_n)\) is the \(n \times n\) diagonal matrix with \(a_1,\ldots,a_n\) as the diagonal elements.
From this representation, it is clear that \(\FPar\) can be constructed from the the first \(\Fr\) columns of \(\IPar\).
\end{proof}

It only remains to show that the construction specified in \cref{sec:general-details} satisfies condition (2) of \cref{thm:construction-conditions}, that is, that \(\IPar\) and \(\FPar\) are superregular.
To do this, we consider the minors of \(\IPar\) and \(\FPar\) as polynomials on \(\theta\).
We show that, due to the structure of the the matrices \(\IPar\) and \(\FPar\) as specified in \cref{sec:general-details}, none of these polynomials can have \(\theta\) as a root as long as the field size is sufficiently large.
Therefore none of the minors can be zero.

\begin{lemma}\label{thm:condition-two}
    Let \(\IPar, \FPar\) be as defined in \cref{sec:general-details}.
    Then \(\IPar\) and \(\FPar\) are superregular, for sufficiently large field size.
\end{lemma}

\begin{proof}
    Let \(\mathbf{R}\) be a \(t \times t\) submatrix of \(\IPar\) or \(\FPar\), determined by the row indices \(i_1 < i_2 < \cdots < i_t\) and the column indices \(j_1 < j_2 < \cdots < j_t\), and denote entry \((i, j)\) of \(\mathbf{R}\) as \(\mathbf{R}[i, j]\).
    The determinant of \(\mathbf{R}\) is defined by the Leibniz formula:
    \begin{gather}
        \det(\mathbf{R})
        = \sum_{\sigma \in \Perm(t)} \Sign(\sigma) \prod_{l = 1}^{t} \mathbf{R}[l, \sigma(l)]
        = \sum_{\sigma \in \Perm(t)} \Sign(\sigma) \theta^{E_\sigma} \label{eq:det-R}\\
        \text{where}\qquad E_\sigma = \sum_{l = 1}^{t} (i_l - 1) (j_{\sigma(l)} - 1),
    \end{gather}
    \(\Perm(t)\) is the set of all permutations on \(t\) elements, and \(\Sign(\sigma) \in \{-1, 1\}\) is the sign of the permutation \(\sigma \in \Perm(t)\) (the sign of a permutation \(\sigma\) depends on the number of inversions in \(\sigma\)).
    Clearly, \(\det(\mathbf{R})\) defines a univariate polynomial \(f_{\mathbf{R}} \in \field{p}[\theta]\).
    We will now show that \(\deg(f_{\mathbf{R}}) = \sum_{l = 1}^{t} (i_l - 1)(j_l - 1)\) by showing that there is a unique permutation \(\sigma^* \in \Perm(t)\) for which \(E_{\sigma^*}\) achieves this value, and that this is the maximum over all permutations in \(\Perm(t)\).
    This means that \(f_{\mathbf{R}}\) has a leading term of degree \(E_{\sigma^*}\).
    
    To prove this, we show that any permutation \(\sigma \in \Perm(t) \backslash \{\sigma^*\}\) can be modified into a permutation \(\sigma'\) such that \(E_{\sigma'} > E_{\sigma}\).
    Specifically, we show that \(\sigma^* = \Id\), the identity permutation.
    Consider \(\sigma \in \Perm(t) \backslash \{\Id\}\): let \(a\) be the smallest index such that \(\sigma(a) \neq a\), let \(b = \sigma^{-1}(a)\), and let \(c = \sigma(a)\).
    Let \(\sigma'\) be such that \(\sigma'(a) = a\), \(\sigma'(b) = c\), and \(\sigma'(d) = \sigma(d)\) for \(d \in [t] \backslash \{a, b\}\).
    In other words, \(\sigma'\) is the result of ``swapping'' the images of \(a\) and \(b\) in \(\sigma\).
    Notice that \(a < b\) and \(a < c\).
    Then, we have that:
    \begin{align}
        E_{\sigma'} - E_\sigma
        &= (i_{a} - 1)(j_{a} - 1) + (i_{b} - 1)(j_{c} - 1) - (i_{a} - 1)(j_{c} - 1) - (i_{b} - 1)(j_{a} - 1) \\
        &= (i_{b} - i_{a})(j_{c} - j_{a}) > 0
    \end{align}
    The last inequality comes from the fact that \(a < b\) implies \(i_a < i_b\) and \(a < c\) implies \(j_a < j_c\).
    Therefore, \(\deg(f_{\mathbf{R}}) = \max_{\sigma \in \Perm(t)} E_\sigma = E_{\Id}\).
    
    Let \(E^*(\Cs, \Ik, \Ir, \Fr)\) be the maximum degree of \(f_{\mathbf{R}}\) over all submatrices \(\mathbf{R}\) of \(\IPar\) or \(\FPar\).
    Then, \(E^*(\Cs,\Ik,\Ir,\Fr)\) corresponds to the diagonal with the largest elements in \(\IPar\) or \(\FPar\).
    In \(\FPar\) this is the diagonal of the square submatrix formed by the bottom \(\Fr\) rows.
    In \(\IPar\) it can be either the diagonal of the square submatrix formed by the bottom \(\Ir\) rows, or by the right \(\Ik\) columns.
    Thus, we have that:
    \begin{align*}
        E^*(\Cs, \Ik, \Ir, \Fr) &= \max\left\{ \sum_{i = 0}^{\Fr - 1} i(\Cs\Ik - \Fr + i), \sum_{i=0}^{\Ir - 1} i(\Ik - \Ir  + i), \sum_{i = 0}^{\Ik - 1} i(\Ir - \Ik + i) \right\}\\
        &= (1/6) \cdot \max
        \begin{Bmatrix}
        \Fr(\Fr - 1)(3\Cs\Ik - \Fr - 1),\\
        \Ir(\Ir - 1)(3\Ik - \Ir - 1),\\
        \Ik(\Ik - 1)(3\Ir - \Ik - 1)
        \end{Bmatrix}.
    \end{align*}
    Let \(D = E^*(\Cs, \Ik, \Ir, \Fr) + 1\).
    Then, if \(\det(\mathbf{R}) = 0\) for some submatrix \(\mathbf{R}\), \(\theta\) is a root of \(f_{\mathbf{R}}\), which is a contradiction since \(\theta\) is a primitive element and the minimal polynomial of \(\theta\) over \(\field{q}\) has degree \(D > \deg(f_{\mathbf{R}})\) \cite{TECC78}.
\end{proof}
This construction is practical only for small values of these parameters since the required field size grows rapidly with the lengths of the initial and final codes.
In \cref{sec:simple-hankel} we present practical low-field-size constructions. 

Combining the above results leads to the following key result on the achievability of the lower bounds on \cost derived in Section~\ref{sec:merge}.
\begin{theorem}
    The explicit construction provided in \cref{sec:general-details} yields \optimal\ linear MDS \codenames\ for all parameter values in the \regime.
\end{theorem}
\begin{proof}
    Follows from \cref{thm:construction-conditions}, \cref{thm:condition-one}, and \cref{thm:condition-two}.
\end{proof}

\section{Low field-size constructions based on superregular Hankel arrays}\label{sec:simple-hankel}

In this section we present alternative constructions for \RegimeCodesDefault\ that require a significantly lower (polynomial) field size than the general construction presented in \cref{sec:construction}.

\textbf{Key idea.} The key idea behind our constructions is to take the matrices \(\IPar\) and \(\FPar\) as submatrices from a specially constructed triangular array of the following form:
\begin{equation}\label{eq:hankel-array}
    T_m:
    \begin{array}{cccccc}
        b_1 & b_2 & b_3 & \cdots & b_{m - 1} & b_m \\
        b_2 & b_3 & \cdots & \cdots & b_m\\
        b_3 & \vdots & \adots & \adots\\
        \vdots & \vdots & \adots\\
        b_{m - 1} & b_m\\
        b_m\\
    \end{array}
\end{equation}
such that every submatrix of \(T_m\) is superregular. Here, (1) \(b_1, \ldots, b_m\) are (not necessarily distinct) elements from \(\field{q}\), and (2) \(m\) is at most the field size \(q\).
The array \(T_m\) is said to have Hankel form, which means that \(T_m[i,j] = T_m[i - 1, j + 1]\), for all \(i \in [2, m],\, j \in [m - 1]\).
We denote \(T_m\) a \emph{superregular Hankel array}.
Such an array can be constructed by employing the algorithm proposed in \cite{RG85} (where the algorithm was employed to construct {generalized Cauchy matrices} to yield generalized Reed-Solomon codes).
We note that the algorithm outlined in \cite{RG85} takes the field size \(q\) as input, and generates \(T_q\) as the output.
It is easy to see that \(T_q\) thus generated can be truncated to generate the triangular array \(T_m\) for any \(m \leq q\).

We construct the initial and final codes by taking submatrices \(\IPar\) and \(\FPar\) from superregular Hankel arrays (the submatrices have to be contained in the triangle where the array is defined).
This guarantees that \(\IPar\) and \(\FPar\) are superregular. In addition, we exploit the Hankel form of the array by carefully choosing the submatrices that form \(\IPar\) and \(\FPar\) to ensure that \(\FPar\) is \BlockConstructible{\Fr} from \(\IPar\).
Given the way we construct these matrices and the properties of \(T_m\), all the initial and final codes presented in this subsection are \emph{generalized doubly-extended Reed-Solomon} codes \cite{RG85}.

The above idea yields a sequence of constructions with a tradeoff between the field size and the range of \(\Fr\) supported.
We first present the two constructions at the extreme ends of this tradeoff, which we call \textit{\hankelone} and \textit{\hankeltwo}.
Construction \hankelone, described in \cref{sec:hankelone}, can be applied whenever \(\Fr \leq \lfloor \Ir / \Cs \rfloor\), and requires a field size of \(q \geq \max\{\In - 1, \Fn - 1\}\).
Construction \hankeltwo\ , described in \cref{sec:hankeltwo}, can be applied whenever \(\Fr \leq \Ir - \Cs + 1\), and requires a field size of \(q \geq \Ik\Ir\).
We then discuss the constructions that fall in between these two constructions in the tradeoff between field size and coverage of \(\Fr\) values in \cref{sec:hankelthree}. In \cref{sec:hankelthree} we also provide a discussion on the ability of these constructions to be optimal even when parameters of the final code are a priori unknown. 
Throughout this section we will assume that \(\Cs \leq \Ir \leq \Ik\).
The ideas presented here are still applicable when \(\Ir > \Ik\), but the constructions and analysis change in minor ways.

\subsection{\hankelone\ construction} \label{sec:hankelone}

\hankelone\ construction provides an \optimal\ linear MDS \RegimeCodeDefault\ when \(\Fr \leq \lfloor \Ir / \Cs \rfloor\), and requires a field size of \(q \geq \max\{\In - 1, \Fn - 1\}\).
Notice that this construction has no penalty in terms of field size for \optimal\ \conversion, since it has the same field size requirement as the maximum between a pair of \([\In, \Ik]\) and \([\Fn, \Fk]\) Reed-Solomon codes \cite{TECC78}.
We start by illustrating the construction with an example.

\begin{figure}
    \centering
    \includegraphics[width=.5\textwidth]{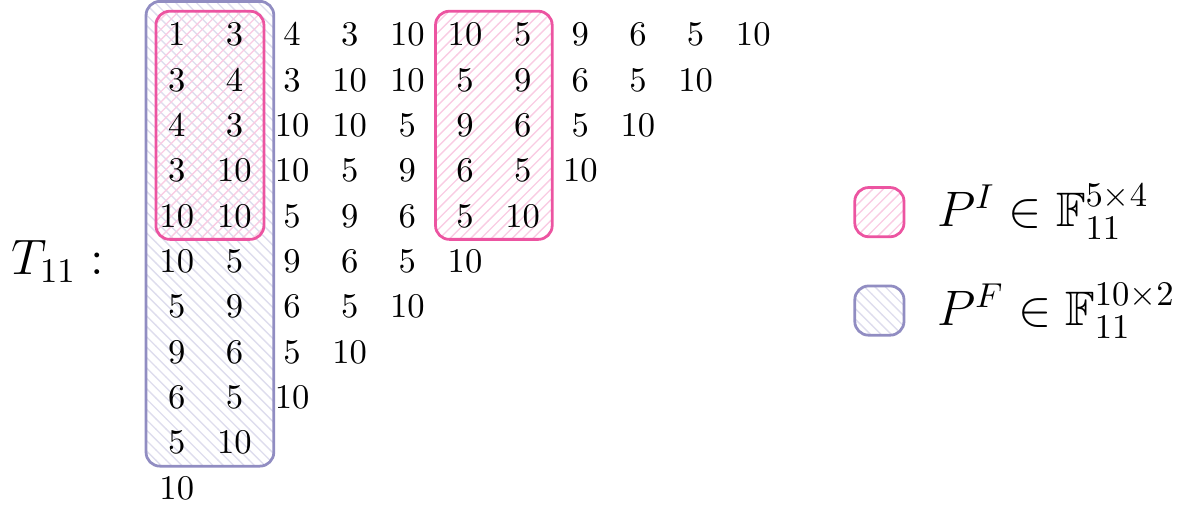}
    \caption{
        \hankelone\ construction parity generator matrices for systematic \ParamCode{9}{5}{12}{10}.
        Notice how matrix \(\FPar\) corresponds to the vertical concatenation of the first two columns and the last two columns of matrix \(\IPar\).
    }
    \label{fig:hankel_ex}
\end{figure}

\begin{example} \label{ex:hankelone}
    Consider the parameters \(\Params{9}{5}{12}{10}\).
    First, we construct a superregular Hankel array of size \(\Fn - 1 = 11\), \(T_{11}\), employing the algorithm in~\cite{RG85}.
    Then choose \(\IPar\) and \(\FPar\) from \(T_{11}\) as shown in \cref{fig:hankel_ex}.
    Checking that these matrices are superregular follows from the superregularity of \(T_{11}\).
    Furthermore, notice that the chosen parity matrices have the following structure:
    \[
    \IPar =
    \begin{bmatrix}
        \vline & \vline & \vline & \vline \\
        \V{p}_1 & \V{p}_2 & \V{p}_3 & \V{p}_4 \\
        \vline & \vline & \vline & \vline \\
    \end{bmatrix}
    \qquad
    \FPar =
    \begin{bmatrix}
        \V{p}_1 & \V{p}_2 \\
        \V{p}_3 & \V{p}_4 \\
    \end{bmatrix}
    \]
    From this structure, it is clear that \(\FPar\) is \BlockConstructible{2} from \(\IPar\).
    The field size required for this construction is \(\Fn - 1 = 11\).
\end{example}

\textbf{General construction.} Now we describe how to construct \(\IPar, \FPar\) for all valid parameters \(\Cs, \Ik, \Ir, \Fr\), where \(\Fr \leq \lfloor \Ir / \Cs \rfloor\).
As seen in \cref{ex:hankelone}, this construction works by splitting the \encvecs\ corresponding to the \(\Ir\) initial parities into \(\Cs\) groups, which are then combined to obtain the (at most) \(\lfloor \Ir / \Cs \rfloor\) new \encvecs.

Let \(T_m\) be as defined in \cref{eq:hankel-array}, with \(m = \Fn - 1\).
Choose \(\FPar\) to be the \(\Fk \times \Fr\) submatrix of the top-left elements of \(T_m\).
Denote the \(\Ik \times ((\Cs - 1)\Ik + \Fr)\) submatrix of the top-left elements of \(T_m\) as \(\textbf{Q}\):
\begin{equation*}
    \FPar =
    \begin{bmatrix}
        b_1 & \cdots & b_{\Fr} \\
        \vdots & \ddots & \vdots \\
        b_{\Cs\Ik} & \cdots & b_{\Cs\Ik + \Fr - 1} \\
    \end{bmatrix}
    \qquad
    \mathbf{Q} =
    \begin{bmatrix}
        b_1 & \cdots & b_{(\Cs - 1)\Ik + \Fr} \\
        \vdots & \ddots & \vdots \\
        b_\Ik & \cdots & b_{\Cs\Ik + \Fr - 1} \\
    \end{bmatrix}
\end{equation*}
We choose \(\IPar\) to be any \(\Ik \times \Ir\) submatrix of \(\mathbf{Q}\) that includes columns \(\{l, \Ik + l, \ldots, (\Cs - 1)\Ik + l\}\).
The Hankel form of array \(T_m\) implies that \(T_m[\Ik(i - 1) + j, l] = T_m[j, \Ik(i - 1) + l]\) for all \(i \in [\Cs],\, j \in [\Ik]\).
As a consequence, we have that the \(l\)-th column of \(\FPar\) is equal to the vertical concatenation of columns \((l, \Ik + l, \ldots, (\Cs - 1)\Ik + l)\) of \(\mathbf{Q}\).

Since both \(\IPar\) and \(\FPar\) are submatrices of \(T_m\), they are superregular.
Furthermore, since every column of \(\FPar\) is the concatenation of \(\Cs\) columns of \(\IPar\), it is clear that \(\FPar\) is \BlockConstructible{\Fr} from \(\IPar\).
Thus \(\IPar\) and \(\FPar\) satisfy both the sufficient properties laid out in \cref{thm:construction-conditions}, and hence \hankelone\ construction is \optimal during \conversion.

\textbf{(Access-optimal) Conversion process.}
During \conversion, the \(\Ik\) data \blocks\ from each of the \(\Cs\) initial stripes remain unchanged, and become the \(\Fk = \Cs\Ik\) data \blocks\ from the final stripe as detailed below.
The \(\Fr\) new (parity) blocks from the final stripe are constructed by accessing \blocks\ from the initial stripes.
To construct the \(l\)-th new \block\ (corresponding to the \(l\)-th column of \(\FPar\), \(l \in [\Fr]\)), read parity \block\ \((i - 1)\Ik + l\) from each initial stripe \(i \in [\Cs]\), and then sum the \(\Cs\) \blocks\ read.
The \encvec\ of the new \block\ will be equal to the sum of the \encvecs\ of the \blocks\ read (recall from \cref{sec:notation} that the initial \encvecs\ are embedded into a \(\Fk\) dimensional space).
This is done for every new \encvec\ \(l \in [\Fr]\).

\subsection{\hankeltwo\ construction} \label{sec:hankeltwo}

\hankeltwo\ construction, in contrast to the \hankelone\ construction above, can handle a broader range of parameter values, at the cost of a slightly larger field-size requirement.
In particular, we present a construction of \optimal\ MDS \RegimeCodeDefault\ for all \(\Fr \leq \Ir - \Cs + 1\), requiring a field size of \(q \geq \Ik\Ir\).
We start with an example illustrating this construction.
 
\begin{figure}
    \centering
    \includegraphics[width=.5\textwidth]{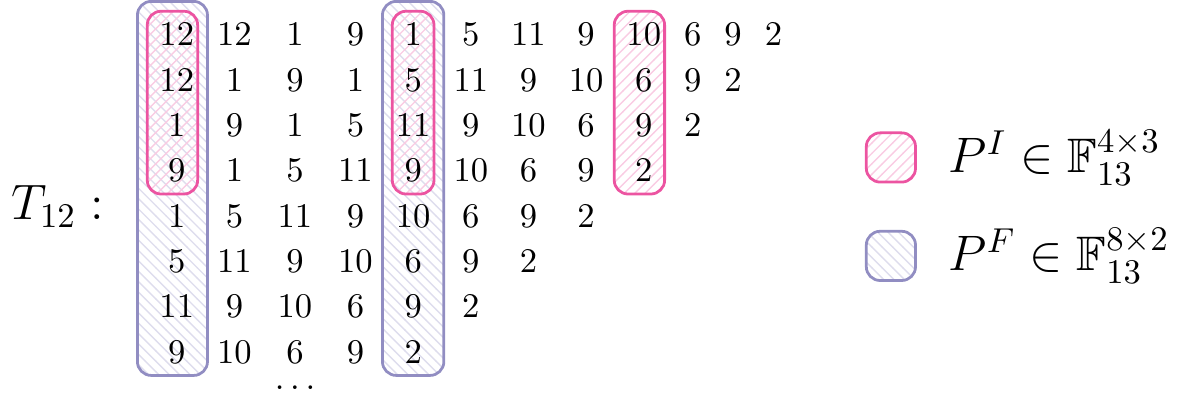}
    \caption{
        \hankeltwo\ construction parity generator matrices for systematic \ParamCode{7}{4}{10}{8}.
        Notice how matrix \(\FPar\) corresponds to the vertical concatenation of the first and second column of \(\IPar\), and the second and third column of \(\IPar\).
    }
    \label{fig:hankel_ex_2}
\end{figure}

\begin{example} \label{ex:hankeltwo}
    Consider parameters \(\Params{7}{4}{10}{8}\).
    First, we construct a superregular Hankel array of size \(\Ik\Ir = 12\), \(T_{12}\), by choosing \(q = 13\) as the field size, and employing the algorithm in \cite{RG85}.
    Then choose \(\IPar\) and \(\FPar\) from \(T_{12}\) as shown in \cref{fig:hankel_ex_2}.
    Both matrices are superregular by the superregularity of \(T_{12}\).
    Notice that the chosen parity matrices have the following structure:
    \[
    \IPar =
    \begin{bmatrix}
        \vline & \vline & \vline \\
        \V{p}_1 & \V{p}_2 & \V{p}_3 \\
        \vline & \vline & \vline \\
    \end{bmatrix}
    \qquad
    \FPar =
    \begin{bmatrix}
        \V{p}_1 & \V{p}_2 \\
        \V{p}_2 & \V{p}_3 \\
    \end{bmatrix}
    \]
    It is easy to see that \(\FPar\) is \BlockConstructible{2} from \(\IPar\).
\end{example}

\textbf{General construction.} Now we describe how to construct \(\IPar\) and \(\FPar\) for all valid parameters \(\Cs, \Ik, \Ir, \Fr\) such that \(\Fr \leq \Ir - \Cs + 1\).
As seen in \cref{ex:hankeltwo}, this construction works by choosing the \(\Ir\) initial parity \encvecs\ so that any \(\Cs\) consecutive initial parity \encvecs\ can be combined into a new \encvec.

Let \(T_m\) be as in \cref{eq:hankel-array}, with \(m \geq \Ik\Ir\).
We take \(\IPar\) and \(\FPar\) as the following submatrices of \(T_m\):
\begin{equation*}
    \IPar =
    \begin{bmatrix}
        b_1 & b_{\Ik + 1} & \cdots & b_{(\Ir - 1)\Ik + 1} \\
        b_2 & b_{\Ik + 2} & \cdots & b_{(\Ir - 1)\Ik + 2} \\
        \vdots & \vdots & \ddots & \vdots \\
        b_{\Ik} & b_{2\Ik} & \cdots & b_{\Ir\Ik} \\
    \end{bmatrix}
    \qquad
    \FPar =
    \begin{bmatrix}
        b_1 & b_{\Ik + 1} & \cdots & b_{(\Fr - 1)\Ik + 1} \\
        b_2 & b_{\Ik + 2} & \cdots & b_{(\Fr - 1)\Ik + 2} \\
        \vdots & \vdots & \ddots & \vdots \\
        b_{\Cs\Ik} & b_{(\Cs + 1)\Ik} & \cdots & b_{(\Cs + \Fr - 1)\Ik} \\
    \end{bmatrix}
\end{equation*}
The Hankel form of array \(T_m\) guarantees that the \(i\)-th column of \(\FPar\) corresponds to the concatenation of columns \((i, i + 1, \ldots, i + \Cs - 1)\) of \(\IPar\).
Thus, \(\FPar\) is \BlockConstructible{\Fr} from \(\IPar\).
Furthermore, since \(\IPar\) and \(\FPar\) are submatrices of \(T_m\), they are superregular.

\textbf{(Access-optimal) Conversion process.} 
During \conversion, the \(\Ik\) data \blocks\ from each of the \(\Cs\) initial stripes remain unchanged, and become the \(\Fk = \Cs\Ik\) data \blocks\ from the final stripe.
The \(\Fr\) new (parity) blocks from the final stripe are constructed by accessing \blocks\ from the initial stripes as detailed below.
To construct the \(l\)-th new \block\ (corresponding to the \(l\)-th column of \(\FPar\), \(l \in [\Fr]\)), read parity \block\ \(l + i - 1\) from each initial stripe \(i \in [\Cs]\), and then sum the \(\Cs\) \blocks\ read. The \encvec\ of the new \block\ will be equal to the sum of the \encvecs\ of the \blocks\ read (recall from \cref{sec:notation} that the initial \encvecs\ are embedded into a \(\Fk\) dimensional space).
This is done for every new \encvec\ \(l \in [\Fr]\).

\subsection{Sequence of Hankel-based constructions and Handling a priori unknown parameters} \label{sec:hankelthree}
\textbf{Sequence of Hankel-based constructions.}
Our idea of Hankel-array-based construction yields a sequence of \optimal MDS \codenames with a tradeoff between field size and the range of \(\Fr\) supported. The two constructions presented in \cref{sec:hankelone} and \cref{sec:hankeltwo} are the two extreme points of this tradeoff.

In particular, our construction can support, for all \(s \in \{\Cs, \Cs + 1, \ldots, \Ir\}\):
\[
    \Fr \leq (s - \Cs + 1)\left\lfloor\frac{\Ir}{s}\right\rfloor
    + \max\{(\Ir \Mod s) - \Cs + 1, 0\},
    \quad
    \text{ for }
    q \geq s\Ik + \left\lfloor\frac{\Ir}{s}\right\rfloor - 1.
\]
The parameter \(s\) corresponds to the number of groups into which the \encvecs\ corresponding to the \(\Ir\) initial parities are split.
That is, each group of consecutive initial parity \encvecs\ has size \(\lfloor \Ir / s \rfloor\) or \(\lceil \Ir / s \rceil\).
The \hankelone\ construction corresponds to \(s = \Cs\) and \hankeltwo\ corresponds to \(s = \Ir\).

\textbf{Handling a priori unknown parameters.} So far, we had assumed that the parameters of the final code, \(\Fn, \Fk\), are known {a priori} and are fixed. As discussed in Section~\ref{sec:conversion}, this is useful in developing an understanding of the fundamental limits of code \conversion. When realizing code \conversion in practice, however, the parameters \(\Fn, \Fk\) might not be known at code construction time (as it depends on the empirically observed failure rates). Thus, it is of interest to be able to \textit{\convert a code optimally to multiple different parameters}. 
The Hankel-array based constructions presented above indeed provide such a flexibility. Our constructions continue to enable \optimal \conversion for any \({\Fk}' = \Cs'\Ik\) and \({\Fn}' = {\Fr}' + {\Fk}'\) with \(0 \leq {\Fr}' \leq \Fr\) and \(2 \leq {\Cs}' \leq \Cs\).

\section{Conclusions and Future directions}\label{sec:conclusion}
In this paper, we propose the ``\codeconversion'' problem, that models the problem of converting data encoded with an \([\In, \Ik]\) code into data encoded with an \([\Fn, \Fk]\) code in a resource-efficient manner.
The proposed problem is motivated by the practical necessity of reducing the overhead of redundancy adaptation in erasure-coded storage systems. This is a new opportunity beckoning coding theorists to enable large-scale real-world storage systems to adapt their redundancy levels to varying failure rates of storage devices, thereby achieving significant savings in resources and energy consumption.
We present the framework of \textit{\codenames} for studying \codeconversions, and fully characterize the fundamental limits for the \cost of \conversions for an important regime of \codenames. Furthermore, we present practical low-field-size constructions for \optimal \codenames  for a wide range of parameters.

This work leads to a number of challenging an potentially impactful open problems. 
An important future direction is to go beyond the \regime considered in this paper and study the fundamental limits on the access cost and construct optimal \codenames for general parameter regimes. Another important future direction is to analyze the fundamental limits on the overhead of other cluster resources during \codeconversions, such as network bandwidth, disk IO, and CPU consumption, and construct \codenames optimizing these resources. Note that while the \optimal \codenames, considered in this paper, also reduce the total network bandwidth, disk IO, and CPU overhead during \conversion as compared to the default approach, the overhead on these other resources may not be optimal.

\section*{Acknowledgements}

We thank Michael Rudow for his valuable feedback and helpful comments during the writing of this paper.

\bibliographystyle{IEEEtran}
\bibliography{IEEEabrv,main}

\begin{thebibliography}{10}
\providecommand{\url}[1]{#1}
\csname url@samestyle\endcsname
\providecommand{\newblock}{\relax}
\providecommand{\bibinfo}[2]{#2}
\providecommand{\BIBentrySTDinterwordspacing}{\spaceskip=0pt\relax}
\providecommand{\BIBentryALTinterwordstretchfactor}{4}
\providecommand{\BIBentryALTinterwordspacing}{\spaceskip=\fontdimen2\font plus
\BIBentryALTinterwordstretchfactor\fontdimen3\font minus
  \fontdimen4\font\relax}
\providecommand{\BIBforeignlanguage}[2]{{%
\expandafter\ifx\csname l@#1\endcsname\relax
\typeout{** WARNING: IEEEtran.bst: No hyphenation pattern has been}%
\typeout{** loaded for the language `#1'. Using the pattern for}%
\typeout{** the default language instead.}%
\else
\language=\csname l@#1\endcsname
\fi
#2}}
\providecommand{\BIBdecl}{\relax}
\BIBdecl

\bibitem{ford2010availability}
D.~Ford, F.~Labelle, F.~Popovici, M.~Stokely, V.~Truong, L.~Barroso, C.~Grimes,
  and S.~Quinlan, ``Availability in globally distributed storage systems,'' in
  \emph{USENIX Symposium on Operating Systems Design and Implementation}, 2010.

\bibitem{rashmi2013hotstorage}
K.~V. Rashmi, N.~B. Shah, D.~Gu, H.~Kuang, D.~Borthakur, and K.~Ramchandran,
  ``A solution to the network challenges of data recovery in erasure-coded
  distributed storage systems: A study on the {F}acebook warehouse cluster,''
  in \emph{Proceedings of USENIX HotStorage}, Jun. 2013.

\bibitem{rashmi2014hitchhiker}
------, ``A {H}itchhiker's guide to fast and efficient data reconstruction in
  erasure-coded data centers,'' in \emph{ACM SIGCOMM}, 2014.

\bibitem{asterisxoring}
M.~Sathiamoorthy, M.~Asteris, D.~Papailiopoulos, A.~G. Dimakis, R.~Vadali,
  S.~Chen, and D.~Borthakur, ``{XORing} elephants: Novel erasure codes for big
  data,'' in \emph{VLDB Endowment}, 2013.

\bibitem{ghemawat2003google}
S.~Ghemawat, H.~Gobioff, and S.~Leung, ``The {Google} file system,'' in
  \emph{ACM SIGOPS Operating Systems Review}, vol.~37, no.~5.\hskip 1em plus
  0.5em minus 0.4em\relax ACM, 2003, pp. 29--43.

\bibitem{facebookECsavings2010_forACM}
\BIBentryALTinterwordspacing
D.~Borthakur, R.~Schmidt, R.~Vadali, S.~Chen, and P.~Kling, ``{HDFS RAID -
  Facebook}.'' [Online]. Available:
  \url{http://www.slideshare.net/ydn/hdfs-raid-facebook}
\BIBentrySTDinterwordspacing

\bibitem{huang2012erasure}
C.~Huang, H.~Simitci, Y.~Xu, A.~Ogus, B.~Calder, P.~Gopalan, J.~Li, and
  S.~Yekhanin, ``Erasure coding in {Windows Azure} storage,'' in
  \emph{Proceedings of USENIX Annual Technical Conference (ATC)}, 2012.

\bibitem{hadoophdfsec}
\BIBentryALTinterwordspacing
{Apache Software Foundation}, ``Apache hadoop: {HDFS} erasure coding,''
  accessed: 2019-07-23. [Online]. Available:
  \url{https://hadoop.apache.org/docs/r3.0.0/hadoop-project-dist/hadoop-hdfs/HDFSErasureCoding.html}
\BIBentrySTDinterwordspacing

\bibitem{HEART}
S.~Kadekodi, K.~V. Rashmi, and G.~R. Ganger, ``Cluster storage systems gotta
  have {HeART}: improving storage efficiency by exploiting disk-reliability
  heterogeneity,'' \emph{{USENIX FAST}}, 2019.

\bibitem{plank05}
J.~Plank, ``T1: Erasure codes for storage applications,'' \emph{Proceedings of
  the 4th USENIX Conference on File and Storage Technologies}, pp. 1--74, 01
  2005.

\bibitem{patterson1988case}
D.~A. Patterson, G.~Gibson, and R.~H. Katz, \emph{A case for redundant arrays
  of inexpensive disks (RAID)}.\hskip 1em plus 0.5em minus 0.4em\relax ACM,
  1988, vol.~17, no.~3.

\bibitem{TECC78}
F.~MacWilliams and N.~Sloane, \emph{The Theory of Error-Correcting Codes},
  2nd~ed.\hskip 1em plus 0.5em minus 0.4em\relax North-holland Publishing
  Company, 1978.

\bibitem{BBBM95}
M.~{Blaum}, J.~{Brady}, J.~{Bruck}, and {Jai Menon}, ``{EVENODD}: an efficient
  scheme for tolerating double disk failures in {RAID} architectures,''
  \emph{IEEE Transactions on Computers}, vol.~44, no.~2, pp. 192--202, feb
  1995.

\bibitem{xu1999x}
L.~Xu and J.~Bruck, ``{X-code}: {MDS} array codes with optimal encoding,''
  \emph{IEEE Transactions on Information Theory}, vol.~45, no.~1, pp. 272--276,
  1999.

\bibitem{huang2008star}
C.~Huang and L.~Xu, ``{STAR}: An efficient coding scheme for correcting triple
  storage node failures,'' \emph{IEEE Transactions on Computers}, vol.~57,
  no.~7, pp. 889--901, 2008.

\bibitem{hafner2005weaver}
J.~L. Hafner, ``{WEAVER} codes: Highly fault tolerant erasure codes for storage
  systems,'' in \emph{Proceedings of the 4th Conference on USENIX Conference on
  File and Storage Technologies - Volume 4}, ser. {FAST'05}.\hskip 1em plus
  0.5em minus 0.4em\relax Berkeley, CA, USA: {USENIX} Association, 2005, pp.
  16--16.

\bibitem{DGWWR10}
A.~G. {Dimakis}, P.~B. {Godfrey}, Y.~{Wu}, M.~J. {Wainwright}, and
  K.~{Ramchandran}, ``Network coding for distributed storage systems,''
  \emph{IEEE Transactions on Information Theory}, vol.~56, no.~9, pp.
  4539--4551, sep 2010.

\bibitem{rashmi2011optimal}
K.~V. Rashmi, N.~B. Shah, and P.~V. Kumar, ``Optimal exact-regenerating codes
  for distributed storage at the {MSR} and {MBR} points via a product-matrix
  construction,'' \emph{IEEE Transactions on Information Theory}, vol.~57,
  no.~8, pp. 5227--5239, 2011.

\bibitem{shah2011interference}
N.~B. Shah, K.~V. Rashmi, P.~V. Kumar, and K.~Ramchandran, ``Interference
  alignment in regenerating codes for distributed storage: Necessity and code
  constructions,'' \emph{IEEE Transactions on Information Theory}, vol.~58,
  no.~4, pp. 2134--2158, 2011.

\bibitem{suh2011journal}
C.~Suh and K.~Ramchandran, ``Exact-repair {MDS} code construction using
  interference alignment,'' \emph{IEEE Transactions on Information Theory}, pp.
  1425--1442, Mar. 2011.

\bibitem{tamo2013zigzag}
I.~Tamo, Z.~Wang, and J.~Bruck, ``Zigzag codes: {MDS} array codes with optimal
  rebuilding,'' \emph{IEEE Transactions on Information Theory}, vol.~59, no.~3,
  pp. 1597--1616, 2013.

\bibitem{cadambe2013asymptotic}
V.~R. Cadambe, S.~A. Jafar, H.~Maleki, K.~Ramchandran, and C.~Suh, ``Asymptotic
  interference alignment for optimal repair of {MDS} codes in distributed
  storage,'' \emph{IEEE Transactions on Information Theory}, vol.~59, no.~5,
  pp. 2974--2987, 2013.

\bibitem{ye2016explicit}
M.~{Ye} and A.~{Barg}, ``Explicit constructions of optimal-access {MDS} codes
  with nearly optimal sub-packetization,'' \emph{IEEE Transactions on
  Information Theory}, vol.~63, no.~10, pp. 6307--6317, oct 2017.

\bibitem{papailiopoulos2013repairTransactions}
D.~Papailiopoulos, A.~Dimakis, and V.~Cadambe, ``Repair optimal erasure codes
  through {H}adamard designs,'' \emph{IEEE Transactions on Information Theory},
  vol.~59, no.~5, pp. 3021--3037, May 2013.

\bibitem{goparaju2017minimum}
S.~Goparaju, A.~Fazeli, and A.~Vardy, ``Minimum storage regenerating codes for
  all parameters,'' \emph{{IEEE} Transactions on Information Theory}, vol.~63,
  no.~10, pp. 6318--6328, 2017.

\bibitem{chowdhury2018newconstructions}
A.~Chowdhury and A.~Vardy, ``New constructions of {MDS} codes with
  asymptotically optimal repair,'' in \emph{2018 {IEEE} International Symposium
  on Information Theory}, 2018, pp. 1944--1948.

\bibitem{mahdaviani2018bandwidth}
K.~Mahdaviani, A.~Khisti, and S.~Mohajer, ``Bandwidth adaptive \& error
  resilient {MBR} exact repair regenerating codes,'' \emph{IEEE Transactions on
  Information Theory}, vol.~65, no.~5, pp. 2736--2759, 2018.

\bibitem{mahdaviani2018product}
K.~Mahdaviani, S.~Mohajer, and A.~Khisti, ``Product matrix {MSR} codes with
  bandwidth adaptive exact repair,'' \emph{IEEE Transactions on Information
  Theory}, vol.~64, no.~4, pp. 3121--3135, 2018.

\bibitem{shah2010flexible}
N.~B. Shah, K.~Rashmi, and P.~V. Kumar, ``A flexible class of regenerating
  codes for distributed storage,'' in \emph{2010 IEEE International Symposium
  on Information Theory}.\hskip 1em plus 0.5em minus 0.4em\relax IEEE, 2010,
  pp. 1943--1947.

\bibitem{shum2011cooperative}
K.~W. Shum, ``Cooperative regenerating codes for distributed storage systems,''
  in \emph{2011 IEEE International Conference on Communications (ICC)}.\hskip
  1em plus 0.5em minus 0.4em\relax IEEE, 2011, pp. 1--5.

\bibitem{abdrashitov2017storage}
V.~Abdrashitov, N.~Prakash, and M.~M{\'e}dard, ``The storage vs repair
  bandwidth trade-off for multiple failures in clustered storage networks,'' in
  \emph{2017 IEEE Information Theory Workshop (ITW)}.\hskip 1em plus 0.5em
  minus 0.4em\relax IEEE, 2017, pp. 46--50.

\bibitem{tamo2014access}
I.~Tamo, Z.~Wang, and J.~Bruck, ``Access versus bandwidth in codes for
  storage,'' \emph{IEEE Transactions on Information Theory}, vol.~60, no.~4,
  pp. 2028--2037, 2014.

\bibitem{goparaju2014improved}
S.~Goparaju, I.~Tamo, and R.~Calderbank, ``An improved sub-packetization bound
  for minimum storage regenerating codes,'' \emph{IEEE Transactions on
  Information Theory}, vol.~60, no.~5, pp. 2770--2779, 2014.

\bibitem{balaji2018tight}
S.~Balaji and P.~V. Kumar, ``A tight lower bound on the sub-packetization level
  of optimal-access {MSR} and {MDS} codes,'' in \emph{2018 IEEE International
  Symposium on Information Theory (ISIT)}.\hskip 1em plus 0.5em minus
  0.4em\relax IEEE, 2018, pp. 2381--2385.

\bibitem{alrabiah2019exponential}
O.~Alrabiah and V.~Guruswami, ``An exponential lower bound on the
  sub-packetization of {MSR} codes,'' in \emph{Proceedings of the 51st Annual
  ACM SIGACT Symposium on Theory of Computing}, ser. STOC 2019.\hskip 1em plus
  0.5em minus 0.4em\relax New York, NY, USA: ACM, 2019, pp. 979--985.

\bibitem{rashmi2013piggybacking}
K.~V. Rashmi, N.~B. Shah, and K.~Ramchandran, ``A piggybacking design framework
  for read-and download-efficient distributed storage codes,'' in \emph{2013
  IEEE International Symposium on Information Theory}, 2013.

\bibitem{rashmi2017piggybacking}
------, ``A piggybacking design framework for read-and download-efficient
  distributed storage codes,'' \emph{IEEE Transactions on Information Theory},
  vol.~63, no.~9, pp. 5802--5820, 2017.

\bibitem{guruswami2017mds}
V.~Guruswami and A.~S. Rawat, ``{MDS} code constructions with small
  sub-packetization and near-optimal repair bandwidth,'' in \emph{ACM-SIAM
  Symposium on Discrete Algorithms}, 2017.

\bibitem{shanmugam2014repair}
K.~Shanmugam, D.~S. Papailiopoulos, A.~G. Dimakis, and G.~Caire, ``A repair
  framework for scalar {MDS} codes,'' \emph{IEEE Journal on Selected Areas in
  Communications}, vol.~32, no.~5, pp. 998--1007, 2014.

\bibitem{guruswami2016repairing}
V.~Guruswami and M.~Wootters, ``Repairing {Reed-Solomon} codes,'' in \emph{ACM
  Symposium on Theory of Computing}, 2016, pp. 216--226.

\bibitem{dau2018repairing}
H.~Dau, I.~M. Duursma, H.~M. Kiah, and O.~Milenkovic, ``Repairing
  {Reed-Solomon} codes with multiple erasures,'' \emph{IEEE Transactions on
  Information Theory}, vol.~64, no.~10, pp. 6567--6582, 2018.

\bibitem{gopalan2012locality}
P.~Gopalan, C.~Huang, H.~Simitci, and S.~Yekhanin, ``On the locality of
  codeword symbols,'' \emph{IEEE Transactions on Information Theory}, vol.~58,
  no.~11, pp. 6925--6934, 2012.

\bibitem{papailiopoulos2014locally}
D.~S. Papailiopoulos and A.~G. Dimakis, ``Locally repairable codes,''
  \emph{IEEE Transactions on Information Theory}, vol.~60, no.~10, pp.
  5843--5855, 2014.

\bibitem{tamo2014family}
I.~Tamo and A.~Barg, ``A family of optimal locally recoverable codes,''
  \emph{IEEE Transactions on Information Theory}, vol.~60, no.~8, pp.
  4661--4676, 2014.

\bibitem{kamath2014codes}
G.~M. Kamath, N.~Prakash, V.~Lalitha, and P.~V. Kumar, ``Codes with local
  regeneration and erasure correction,'' \emph{IEEE Transactions on Information
  Theory}, vol.~60, no.~8, pp. 4637--4660, 2014.

\bibitem{cadambe2015bounds}
V.~R. Cadambe and A.~Mazumdar, ``Bounds on the size of locally recoverable
  codes,'' \emph{IEEE Transactions on Information Theory}, vol.~61, no.~11, pp.
  5787--5794, 2015.

\bibitem{tamo2016optimal}
I.~Tamo, D.~S. Papailiopoulos, and A.~G. Dimakis, ``Optimal locally repairable
  codes and connections to matroid theory,'' \emph{IEEE Transactions on
  Information Theory}, vol.~62, no.~12, pp. 6661--6671, 2016.

\bibitem{tamo2016bounds}
I.~Tamo, A.~Barg, and A.~Frolov, ``Bounds on the parameters of locally
  recoverable codes,'' \emph{IEEE Transactions on Information Theory}, vol.~62,
  no.~6, pp. 3070--3083, 2016.

\bibitem{barg2017locally}
A.~Barg, K.~Haymaker, E.~W. Howe, G.~L. Matthews, and A.~V{\'a}rilly-Alvarado,
  ``Locally recoverable codes from algebraic curves and surfaces,'' in
  \emph{Algebraic Geometry for Coding Theory and Cryptography}.\hskip 1em plus
  0.5em minus 0.4em\relax Springer, 2017, pp. 95--127.

\bibitem{agarwal2018combinatorial}
A.~Agarwal, A.~Barg, S.~Hu, A.~Mazumdar, and I.~Tamo, ``Combinatorial
  alphabet-dependent bounds for locally recoverable codes,'' \emph{IEEE
  Transactions on Information Theory}, vol.~64, no.~5, pp. 3481--3492, 2018.

\bibitem{mazumdar2018capacity}
A.~{Mazumdar}, ``Capacity of locally recoverable codes,'' in \emph{2018 IEEE
  Information Theory Workshop}, nov 2018, pp. 1--5.

\bibitem{rashmi2011enabling}
K.~V. Rashmi, N.~B. Shah, and P.~V. Kumar, ``Enabling node repair in any
  erasure code for distributed storage,'' in \emph{2011 IEEE International
  Symposium on Information Theory Proceedings}.\hskip 1em plus 0.5em minus
  0.4em\relax IEEE, 2011, pp. 1235--1239.

\bibitem{MZT18}
S.~{Mousavi}, T.~{Zhou}, and C.~{Tian}, ``Delayed parity generation in {MDS}
  storage codes,'' in \emph{2018 IEEE International Symposium on Information
  Theory (ISIT)}, Jun. 2018, pp. 1889--1893.

\bibitem{XSBP15}
M.~Xia, M.~Saxena, M.~Blaum, and D.~A. Pease, ``A tale of two erasure codes in
  {HDFS},'' in \emph{Proceedings of the 13th USENIX Conference on File and
  Storage Technologies}, ser. {FAST'15}.\hskip 1em plus 0.5em minus 0.4em\relax
  Berkeley, CA, USA: USENIX Association, 2015, pp. 213--226.

\bibitem{konwar2017layered}
K.~M. Konwar, N.~Prakash, N.~Lynch, and M.~M{\'e}dard, ``A layered architecture
  for erasure-coded consistent distributed storage,'' in \emph{Proceedings of
  the ACM Symposium on Principles of Distributed Computing}.\hskip 1em plus
  0.5em minus 0.4em\relax ACM, 2017, pp. 63--72.

\bibitem{cadambe2018ares}
V.~Cadambe, N.~Nicolaou, K.~M. Konwar, N.~Prakash, N.~Lynch, and M.~M{\'e}dard,
  ``{ARES}: Adaptive, reconfigurable, erasure coded, atomic storage,'' 2018, to
  appear in 39th IEEE International Conference on Distributed Computing Systems
  (ICDCS 2019), Dallas, Texas, July 2019.

\bibitem{wang2017multi}
Z.~Wang and V.~R. Cadambe, ``Multi-version coding: An information-theoretic
  perspective of consistent distributed storage,'' \emph{IEEE Transactions on
  Information Theory}, vol.~64, no.~6, pp. 4540--4561, 2017.

\bibitem{GRS06}
H.~{Gluesing-Luerssen}, J.~{Rosenthal}, and R.~{Smarandache}, ``Strongly-{MDS}
  convolutional codes,'' \emph{IEEE Transactions on Information Theory},
  vol.~52, no.~2, pp. 584--598, feb 2006.

\bibitem{RG85}
R.~M. Roth and G.~Seroussi, ``On generator matrices of {MDS} codes,''
  \emph{IEEE Transactions on Infortmation Theory}, vol.~31, no.~6, pp.
  826--830, Nov. 1985.

\end{thebibliography}

\end{document}